\documentclass[a4paper,UKenglish,cleveref, autoref, thm-restate]{lipics-v2021}

\usepackage{wrapfig}

\usepackage{graphicx}
\usepackage{tikzsymbols}

\usepackage[linesnumbered,lined,boxed,commentsnumbered]{algorithm2e}
\usepackage{multirow}
\usepackage{tabularx}
\usepackage{tikz}
\usetikzlibrary{calc}
\usepackage{multirow}

\usepackage[]{color}
\usepackage{bbding}
\usepackage{booktabs}

\title{Learning-Augmented Online TSP on Rings, Trees, Flowers and (almost) Everywhere Else} 

\titlerunning{Learning-Augmented Online TSP on Rings, Trees, Flowers and Everywhere Else} 
\newcommand{\lipsix}{Sorbonne Universit\'e, CNRS, LIP6, F-75005 Paris, France}

\author{Evripidis Bampis}{\lipsix}{evripidis.bampis@lip6.fr}{https://orcid.org/0000-0002-4498-3040}{}
\author{Bruno Escoffier}{\lipsix \and Institut Universitaire de France, Paris, France}{bruno.escoffier@lip6.fr}{https://orcid.org/0000-0002-6477-8706}{}
\author{Themis Gouleakis}{National University of Singapore, Singapore, Singapore}{tgoule@nus.edu.sg}{https://orcid.org/0000-0002-4056-0489}{}
\author{Niklas Hahn}{\lipsix}{niklas.hahn@lip6.fr}{https://orcid.org/0000-0002-4929-0542}{}
\author{Kostas Lakis}{ETH Zürich, Zürich, Switzerland}{klakis@student.ethz.ch}{https://orcid.org/0009-0004-5595-1839}{}
\author{Golnoosh Shahkarami}{Max-Planck-Institut für Informatik, Universität des Saarlandes, Saarbrücken, Germany}{gshahkar@mpi-inf.mpg.de}{https://orcid.org/0000-0002-6169-7337}{}
\author{Michalis Xefteris}{\lipsix}{michail.xefteris@lip6.fr}{https://orcid.org/0009-0006-2894-3029}{}

\authorrunning{Bampis et al.} 

\Copyright{Evripidis Bampis, Bruno Escoffier, Themis Gouleakis, Niklas Hahn, Kostas Lakis, Golnoosh Shahkarami, and Michalis Xefteris} 



\category{} 

\relatedversion{} 

\funding{This work was partially funded by the grant ANR-19-CE48-0016 from the French National Research Agency (ANR).}

\acknowledgements{
}

\nolinenumbers 

\hideLIPIcs 

\allowdisplaybreaks

\newcommand{\opt}{\ensuremath{\operatorname{OPT}}}
\newcommand{\alg}{\ensuremath{\operatorname{ALG}}}
\newcommand{\ratio}{\ensuremath{\rho}}
\newcommand{\orig}{\ensuremath{\mathcal{O}}}
\newcommand{\oracle}{\ensuremath{\mathcal{D}}}
\newcommand{\firingPerm}{\ensuremath{\sigma_0}}
\newcommand{\chosenPerm}{\ensuremath{\sigma_1}}
\newcommand{\algoStart}{\ensuremath{T}}
\newcommand{\tree}{\ensuremath{\mathcal{T}}}

\newcommand{\swag}{\hyperref[algo:general]{\texttt{SWAG}}}
\newcommand{\laswag}{\hyperref[algo:generalSmoothRobust]{\texttt{LA-SWAG}}}

\newcommand{\arcThroughThreePoints}[4][]{
\coordinate (middle1) at ($(#2)!.5!(#3)$);
\coordinate (middle2) at ($(#3)!.5!(#4)$);
\coordinate (aux1) at ($(middle1)!1!90:(#3)$);
\coordinate (aux2) at ($(middle2)!1!90:(#4)$);
\coordinate (center) at ($(intersection of middle1--aux1 and middle2--aux2)$);
\draw[#1, red, dashed] 
 let \p1=($(#2)-(center)$),
      \p2=($(#4)-(center)$),
      \n0={veclen(\p1)},       
      \n1={atan2(\y1,\x1)}, 
      \n2={atan2(\y2,\x2)},
      \n3={\n2>\n1?\n2:\n2+360}
    in (#2) arc(\n1:\n3:\n0);
}

\begin{document}
	
	\maketitle
        
	\begin{abstract}
        We study the Online Traveling Salesperson  Problem (OLTSP) with predictions. In OLTSP, a sequence of initially unknown requests arrive over time at points (locations) of a metric space. The goal is, starting from a particular point of the metric space (the origin), to serve all these requests while minimizing the total time spent. 
        The server moves with unit speed or is ``waiting'' (zero speed) at some location.
        We consider two variants: in the open variant, the goal is achieved when the last request is served. In the closed one, the server additionally has to return to the origin. We adopt a prediction model, introduced for OLTSP on the line~\cite{GouleakisLS22}, in which the predictions correspond to the locations of the requests and extend it to more general metric spaces.
        
        We first propose an oracle-based algorithmic framework, inspired by previous work~\cite{tsp_l}. This framework allows us to design online algorithms for general metric spaces that provide competitive ratio guarantees which, given perfect predictions, beat the best possible classical guarantee (\emph{consistency}). Moreover, they degrade gracefully along with the increase in error (\emph{smoothness}), but always within a constant factor of the best known competitive ratio in the classical case (\emph{robustness}). 
        
        Having reduced the problem to designing suitable efficient oracles, we describe how to achieve this for general metric spaces as well as specific metric spaces (rings, trees and flowers), the resulting algorithms being tractable in the latter case. The consistency guarantees of our algorithms are tight in almost all cases, and their smoothness guarantees only suffer a linear dependency on the error, which we show is necessary. Finally, we provide robustness guarantees improving previous results.
	\end{abstract}

    \clearpage
    \addtocounter{page}{-1}
 
    \section{Introduction}\label{sec:intro}

In the classical Traveling Salesperson Problem (TSP), we are  given a set of locations as well as the pairwise distances between them and the objective is to find a shortest tour visiting all the locations. TSP is one of the most fundamental and well studied problems in Computer Science~\cite{Lawler91}. We focus on the online version of the problem in a metric space, the Online Traveling Salesperson Problem (OLTSP), introduced in the seminal paper of Ausiello et al.~\cite{AusielloFLST01}.  In OLTSP, the input arrives over time, i.e., new requests (locations) that have to be visited by the traveler (or server) will appear during the travel. The time in which a request is communicated to the traveler  is called its release time (or release date). The objective is the minimization of the total traveled time assuming that at any time the traveler either moves at unit speed or is ``waiting'' (zero speed) at some location\footnote{Note that this choice of possible speeds is w.l.o.g., as any setting where the maximum speed is bounded can be reduced to this setting. Specifically, if the maximum speed is some $S>0$, we can normalize the maximum speed to be $1$ and multiply all distances by $S$. Also, our setting can simulate any other setting where the algorithm is free to chose either the unit speed of some other speed $0<s<1$. Moving at speed $s$ is equivalent to dividing the time into intervals of length $dt\rightarrow 0$ and moving at unit speed only for time $s\cdot dt$ within each interval.}. We study both the {\em closed} variant, where the server is required to return to the origin after serving all requests, and the {\em open} variant, where the server does not have to  return to the origin after serving all the requests.  A series of papers considered many variations of OLTSP in different metric spaces (general metric space~\cite{AusielloFLST01,tsp_l}, the line~\cite{AusielloFLST01,BjeldeHDHLMSSS21,GouleakisLS22}, the semi-line~\cite{AusielloDLP04,tsp_l,BlomKPS01},  the ring~\cite{tsp_l, 10.1007/978-3-030-14812-6_20} and the star~\cite{tsp_l}). 

The motivation of studying OLTSP and its variations comes from applications in many different domains, such as e.g.\ logistics and robotics~\cite{AscheuerGKR90,PsaraftisSMK90}. 
In the framework of competitive  analysis, the performance of an online algorithm is usually evaluated using the competitive ratio which is defined as the maximum ratio between the cost of the online algorithm and the cost of an optimal offline algorithm, which by definition has knowledge of the entire input in advance, over all input instances. However, it is admitted that the competitive analysis approach can be overly pessimistic as it is calculated considering worst-case instances, giving a lot of power to  the adversary. 
Hence, many papers try  to limit the power of the adversary~\cite{BlomKPS01}, 
or give extra knowledge and hence more power to the online algorithm~\cite{AllulliAL05,tsp_l,Jaillet}. 

More recently, the framework of   {\em Learning-Augmented (LA) algorithms} has emerged due to the vibrant successes of Machine Learning methods and Artificial Intelligence in predicting and learning the unknown (i.e., future inputs in the case of online algorithms) based on data~\cite{LykourisV21}. 
In this line of research, the goal is to utilize predictions of the future input that have potentially been acquired using a learning algorithm, in order to have provably improved competitive ratio in the case that the predictions are accurate enough, while maintaining worst-case guarantees even if the prediction error is arbitrarily large. 
In particular, a (possibly erroneous) prediction of the input is given to the algorithm and the goal is to design algorithms with a good performance guarantee when the prediction is accurate (consistency), a not too bad (and bounded) performance when the prediction is wrong (robustness) and a gradual deterioration of the competitive ratio with respect to the prediction error (smoothness). We give more precise definitions in Section \ref{sec:prelim}.

Recent works have proposed a variety of approaches to tackle OLTSP with predictions (LA-OLTSP).
In~\cite{GouleakisLS22}, Gouleakis et al.\ studied a learning-augmented framework for OLTSP on the line. They introduced a prediction model in  which the predictions correspond to the locations of the requests. They proposed LA algorithms for both the closed and the open variants that are consistent, smooth and robust.   
Bampis et al., in~\cite{tsp_l}, considered the case of perfect predictions of the locations and studied different metric spaces (general metric space, semi-line, ring, and star) and proposed competitive online algorithms and lower bounds.
In~\cite{HuWLCL22}, Hu et al.\ proposed three different prediction models for OLTSP. In two of their models, each request is associated to a prediction for both its release time and its location while in the third one  the prediction is just the release time of the last request. In \cite{BernardiniLMMSS22}, Bernardini et al.\ studied OLTSP with predictions of both the release time and the location of each request. They introduced a new error measure, the cover error, and they also considered other online graph problems. More recently, Chawla and Christou~\cite{Chawla23} studied the Online Time-Windows TSP with predictions of release time and location of the requests.

In this work, we adopt the prediction model of~\cite{GouleakisLS22}. We propose a general oracle-based framework that allows us to design consistent, smooth and robust LA 
algorithms for both the closed and open variants of the problem for different metric spaces. We also provide some interesting lower bounds.


\subsection{Our contributions and techniques}
In this paper, we propose a novel approach to improve the competitive ratio of OLTSP using predictions concerning the locations of requests in general metric spaces. Our algorithms provide tight competitive ratio guarantees in most cases. Moreover, we show how to get polynomial-time/FPT algorithms in specific metrics, namely rings, trees and flowers.

Bampis et al.~\cite{tsp_l} gave an algorithm for general metrics with a competitive ratio of $3/2$ for the case of perfect predictions (known locations), which is tight. First, in Section \ref{subsec:perfectpred}, we modify this algorithm (still under the assumption of perfect predictions) and introduce our main oracle-based 3/2-competitive framework, which we call \emph{Strategically Wait And Go} ($\swag$, for pseudocode see Algorithm~\ref{algo:general}). 
The main idea is to consider a suitable subset of permutations of the requests, referred to as \hyperref[Def:DominationPermutation]{\emph{Dominating permutations}}, given by a so-called \hyperref[def:oracle]{\emph{Domination oracle}} instead of all the permutations. This allows for a reduction of the running time, since the bottleneck is located in the cardinality of the set of considered permutations. This restriction of the permutation set preserves the consistency of $3/2$.

Then, we introduce our main algorithm, \emph{Learning-Augmented Strategically Wait And Go} ($\laswag$, for pseudocode see Algorithm~\ref{algo:generalSmoothRobust}), which does not assume perfect predictions. $\laswag$ is consistent, smooth, and robust. More formally, in Section \ref{subsec:smooth} we show the following.

\begin{restatable}[Consistency and Smoothness]{theorem}{ConsistencySmoothness}
\label{th:smoothness}
    $\laswag$ has a competitive ratio of at most $3/2 + 5\eta$ for both closed and open LA-OLTSP.
\end{restatable}
Here, $\eta$ is the error of the prediction (defined formally later) that captures the normalized sum of distances between predicted and actual locations of the requests.  Note that for $\eta = 0$, we get a consistency of $3/2$, which is tight for all cases except the open variant on trees. 

Additionally, we show a smoothness lower bound of $3/2 + \eta/2$ for the open variant with $\eta \in [0,1/3]$ (Proposition \ref{smoothnessLowerBound}), implying that a linear dependency on $\eta$ is required.

Regarding robustness, the algorithm in~\cite{GouleakisLS22} achieves $3$-robustness on the line. 
$\laswag$ improves this bound for general metrics and further so in specific metrics. 

\begin{restatable}[Robustness-Closed]{theorem}{robustnessColsed}
\label{robustness:2.75}
$\laswag$ is $2.75$-robust for closed LA-OLTSP in general metric spaces, and $2.5$-robust in Euclidean spaces and in trees.
\end{restatable}

\begin{restatable}[Robustness-Open]{theorem}{robustnessOpen}
\label{robustness:2.84}
    $\laswag$ is $\left(3 - 1/6 \right)$-robust for open LA-OLTSP in general metric spaces, and  $\left(3 - 1/3\right)$-robust in trees.
\end{restatable}

Our analysis for $2.5$-robustness and $\left(3 - 1/3\right)$-robustness is tight even on the line (Remarks \ref{tight2.5} and \ref{tight3-1/3}).
Moreover, we show a negative result concerning the consistency/robustness trade-off of any algorithm for the open variant (Lemma \ref{tradeoff}).

The main technical contribution of our work, found in Sections \ref{sec:singleExponential} and \ref{sec:polyTimeAlgos}, is the implementation of the domination oracles we have referred to. On a high level, we say that a permutation $\pi_{dom}$ dominates another permutation $\pi$ at time $t$, if the following conditions hold. Assuming $q$ is the first unreleased request in $\pi$ at time $t$, the distance traveled up to $q$ is not longer in $\pi_{dom}$, and also a superset of the requests preceding $q$ in $\pi$ is visited before $q$. Moreover, $\pi_{dom}$ induces a not longer path than $\pi$. These two key facts allow us to preserve $3/2$-consistency. 

For general metrics, we achieve this domination using a very similar idea as the one employed in the definition of the $O(n^22^n)$ dynamic programming solution of the classical TSP~\cite{Bellman, Karp}. That is, for any possible subset of released requests that might have been served by $\pi$ before $q$, we simply build two optimal paths for the parts before and after $q$ (without release times) and then we concatenate them to get a dominating permutation. We call the resulting sets \textit{general dominating sets}. Any permutation is dominated by the one corresponding to the correct guess of requests served up to $q$. This yields a single-exponential time algorithm overall.

\begin{restatable}[General Metrics]{theorem}{expGeneral}
    $\laswag$  with an oracle $\oracle$ which uses the general dominating sets runs in single-exponential time and is $\min\{3/2+5\eta, 2.75\}$-competitive for the closed variant and $\min\{3/2+5\eta, 3-1/6\}$-competitive for the open variant of LA-OLTSP. 
\end{restatable}

The overarching insight behind how we reduce the runtime in specific metrics is the fact that we do not really need to try all possible subsets of requests served before $q$. For example, in trees, we first show a structural result about the optimal solutions (with release times). Specifically, we prove that the requests placed along a path from a leaf to the origin (considered the root) can be assumed to be served in a very specific order. This is the order in which the requests are encountered as one traverses the path from the leaf to the origin. This fact allows us to design a domination oracle which, roughly, provides a single permutation for any subset of \textit{leaves} visited before $q$. Hence, we get an FPT algorithm parameterized by the number $l$ of leaves of the tree.

\begin{restatable}[Trees]{theorem}{trees}
    There exists a Domination oracle for $\laswag$ in trees which yields a time complexity of $O(2^l \cdot n^3)$ for the closed variant and $O(2^l \cdot n^4)$ for the open variant, where $l$ is the number of leaves of the input tree.
\end{restatable}
As a first step towards more general graphs, we deal with the concept of cycles by considering the ring. While we cannot retrieve the exact same structural result about online optimal solutions, we show something quite similar. Namely, we prove that the cyclic nature of the ring may be utilized only once by an optimal solution. After such a cyclic traversal, we can assume that the ring is split in half, yielding a tree which we know how to deal with.
\begin{restatable}[Ring]{theorem}{ring}
    There exists a Domination oracle for $\laswag$ in the ring which yields a time complexity of $O(n^3)$ for the closed variant and $O(n^5)$ for the open variant.
\end{restatable}
Finally, we combine the two previous sets of ideas to tackle flowers. Flowers are essentially comprised of a bunch of rings (petals) and a semi-line (stem), all of which are attached to the origin. It is still true that each single ring may be assumed to be traversed with a loop only once in this case. It turns out that we can consider at most $6$ options for every petal.
\begin{restatable}[Flowers]{theorem}{flowers}
    There exists a Domination oracle for $\laswag$ in flowers which yields a time complexity of $O(6^{p} \cdot n^3)$ for the closed variant and $O(6^{p} \cdot n^5)$ for the open variant, where $p$ is the number of petals of the input flower. 
\end{restatable}

We summarize our results in Tables~\ref{CSRR} and~\ref{Consistency-UB}. In the closed variant, there is a lower bound of $3/2$ even in the case of a line~\cite{GouleakisLS22}. In the open one, we show a lower bound of $\approx 1.468$ even in the case of a line in the Appendix (Proposition~\ref{open_lb_line}) and there is a lower bound of $3/2$ even in the case of a ring~\cite{tsp_l}.

\begin{table}[h] 
\caption{Consistency, smoothness, robustness and runtime guarantees of $\laswag$.}
\label{CSRR}
\centering
\begin{tabular}{|c|c|cc|cc|}
\hline
\multirow{2}{*}{} & \begin{tabular}[c]{@{}c@{}}Smoothness\\ (Consistency for $\eta=0$)\end{tabular} & \multicolumn{2}{c|}{Robustness}    & \multicolumn{2}{c|}{Runtime}                 \\ \cline{2-6} 
                  & Closed/Open                                                                 & \multicolumn{1}{c|}{Closed} & Open & \multicolumn{1}{c|}{Closed}     & Open       \\ \hline
Tree              & \multirow{5}{*}{$3/2+5\eta$}                                                     & \multicolumn{1}{c|}{2.5}    & $3-1/3$ & \multicolumn{1}{c|}{$O(2^l\cdot n^3)$} & $O(2^l \cdot n^4)$ \\ \cline{1-1} \cline{3-6} 
Ring              &                                                                             & \multicolumn{1}{c|}{2.75}   & $3-1/6$ & \multicolumn{1}{c|}{$O(n^3)$}     & $O(n^5)$     \\ \cline{1-1} \cline{3-6} 
Flower            &                                                                             & \multicolumn{1}{c|}{2.75}   & $3-1/6$ & \multicolumn{1}{c|}{$O(6^p \cdot n^3)$}  & $O(6^p \cdot n^5)$  \\ \cline{1-1} \cline{3-6} 
Euclidean         &                                                                             & \multicolumn{1}{c|}{2.5}    & $3-1/6$ & \multicolumn{1}{c|}{$O(n^2 \cdot 2^n)$}  & $O(n^2 \cdot 2^n)$  \\ \cline{1-1} \cline{3-6} 
General           &                                                                             & \multicolumn{1}{c|}{2.75}   & $3-1/6$ & \multicolumn{1}{c|}{$O(n^2 \cdot 2^n)$}  & $O(n^2 \cdot 2^n)$  \\ \hline
\end{tabular}
\end{table}

\begin{table}[h] 
\caption{Consistency of $\laswag$ and other tractable algorithms. The upper bounds with * are given by an FPT algorithm and a polytime algorithm otherwise. Tight bounds are denoted in bold.}
\label{Consistency-UB}
\centering
\begin{tabular}{|c|cccccccc|}
\hline
\multirow{3}{*}{} & \multicolumn{8}{c|}{\begin{tabular}[c]{@{}c@{}}Consistency\end{tabular}}                                       \\ \cline{2-9} 
                  & \multicolumn{4}{c|}{Closed}                                                   & \multicolumn{4}{c|}{Open}                                \\ \cline{2-9} 
                  & \multicolumn{3}{c|}{Previous work} & \multicolumn{1}{c|}{\textbf{This paper}} &
                  \multicolumn{3}{c|}{Previous work} & \textbf{This paper} \\ \hline
Line              & \phantom{[18]} & \multicolumn{1}{c}{$\boldsymbol{3/2}$} & \multicolumn{1}{c|}{\cite{GouleakisLS22}}           & \multicolumn{1}{c|}{$\boldsymbol{3/2\phantom{^*}}$}                 & \phantom{[18]} & \multicolumn{1}{c}{5/3} & \multicolumn{1}{c|}{\cite{GouleakisLS22}}           & $3/2\phantom{^*}$                 \\ \hline
Star             & \phantom{[18]} & \multicolumn{1}{c}{7/4+$\epsilon$} & \multicolumn{1}{c|}{\cite{tsp_l}}         & \multicolumn{1}{c|}{$\boldsymbol{3/2^*}$}                & \phantom{[18]} & \multicolumn{1}{c}{2} & \multicolumn{1}{c|}{\cite{tsp_l}}             & $3/2^*$               \\ \hline
Tree              & \phantom{[18]} & \multicolumn{1}{c}{2} & \multicolumn{1}{c|}{\cite{AusielloFLST01}}             & \multicolumn{1}{c|}{$\boldsymbol{3/2^*}$}                & \phantom{[18]} & \multicolumn{1}{c}{2} & \multicolumn{1}{c|}{\cite{tsp_l}}             & $3/2^*$                \\ \hline
Ring              & \phantom{[18]} & \multicolumn{1}{c}{5/3} & \multicolumn{1}{c|}{\cite{tsp_l}}           & \multicolumn{1}{c|}{$\boldsymbol{3/2\phantom{^*}}$}                 & \phantom{[18]} & \multicolumn{1}{c}{2} &\multicolumn{1}{c|}{\cite{tsp_l}}             & $\boldsymbol{3/2\phantom{^*}}$                 \\ \hline
Flower            & \phantom{[18]} & \multicolumn{1}{c}{2} & \multicolumn{1}{c|}{\cite{AusielloFLST01}}             & \multicolumn{1}{c|}{$\boldsymbol{3/2^*}$}               & \phantom{[18]} & \multicolumn{1}{c}{2} &\multicolumn{1}{c|}{\cite{tsp_l}}             & $\boldsymbol{3/2^*}$                \\ \hline
\end{tabular}
\end{table}

\subsection{Further related works}

The offline version of the problem, in which the locations and release times are known in advance, has been studied in~\cite{BjeldeHDHLMSSS21, PsaraftisSMK90} for both closed and open variants.
For OLTSP, a 2-competitive algorithm for the closed variant and a 2.5-competitive algorithm for the open variant have been proposed in general metric spaces by~\cite{AusielloFLST01}. Specifically on the line, there exist lower bounds of $1.64$~\cite{AusielloFLST01} and $2.04$~\cite{BjeldeHDHLMSSS21} for the closed and open variants, respectively.

In addition to online TSP, there are several works that have explored learning augmented settings. The online caching problem with predictions was investigated by~\cite{LykourisV21}, and the initial results were improved by~\cite{AntoniadisCE0S20, rohatgi2020near, Wei20}. Adopting the LA approach, algorithms were developed for the ski-rental problem~\cite{DBLP:conf/innovations/0001DJKR20, pmlr-v97-gollapudi19a, NEURIPS2018_73a427ba, WangL20} as well as for scheduling problems~\cite{DBLP:journals/corr/abs-2112-03082, BamasMRS20, Mitzenmacher20, 48659}.
There is also literature on learning augmented algorithms for classical data structures~\cite{DBLP:conf/sigmod/KraskaBCDP18}, bloom filters~\cite{NEURIPS2018_0f49c89d}, routing problems~\cite{evripidis2, megow2, lak2},  online selection and matching problems~\cite{AntoniadisGKK20, DuttingLLV21} and a more general framework of online primal-dual algorithms~\cite{BamasMS20}. There is a survey~\cite{MitzenmacherV20} and an updated list of papers~\cite{website} in this area.	
    \section{Preliminaries}
\label{sec:prelim}	

\paragraph*{Online TSP (OLTSP).}
	The input of OLTSP consists of a metric space $M$ with a distinguished point $\orig$ (the origin), and a set $Q=\{q_1,...,q_n\}$ of $n$ requests. Every request $q_i$ is a pair ($x_i, t_i$), where $x_i$ is a point of $M$ and $t_i \ge 0$ is a real number. 
    We use $t$ to quantify time. The number $t_i$ represents the moment after which the request $q_i$ can be served (release time). A server located at the origin at time $t = 0$, which can move with unit speed, must serve all the requests after their release times with the goal of minimizing the total completion time (makespan).

    We consider a wide class of continuous metric spaces $M$ whose corresponding distance metric $d(x,y)$ is defined as the shortest path from $x\in M$ to $y\in M$ and is continuous in $M$, as in~\cite{AusielloFLST01}. We call this class general metric spaces (or general metrics). The release times for continuous metric spaces can be any non-negative real number.
	
	For the rest of the paper, we denote the total completion time of an online algorithm $\alg$ by $|\alg|$ and that of an optimal (offline) solution $\opt$ by $|\opt|$. We recall that an algorithm $\alg$ is $\ratio$-competitive if on all instances we have $|\alg|\leq \ratio \cdot |\opt|$.

\paragraph*{Learning-augmented algorithms.}

In order to measure the quality of the predictions, we will define a prediction error $\eta$.
LA algorithms have three main properties. We use the formal definitions in \cite{GouleakisLS22} here.
We say that an algorithm is 
\begin{itemize}
    \item  $\alpha$-\textit{consistent}, if it is $\alpha$-\textit{competitive} when $ \eta = 0$, 
    \item $\beta$-\textit{robust}, if it is always $\beta$-\textit{competitive} regardless of $\eta$, and 
    \item $\gamma$-\textit{smooth} for a continuous function $\gamma(\eta)$, if it is $\gamma(\eta)$-\textit{competitive}.
\end{itemize}
In general, if $c$ is the best competitive ratio achievable without predictions, it is desirable to have $\alpha < c$, $\beta \le k \cdot c$ for some constant $k$ and also the function $\gamma$ should increase from $\alpha$ to $\beta$ along with the error $\eta$.

\paragraph*{Our prediction setting.}
Let $Q=\{q_1,\dots,q_n\}$ be the set of requests. As we mentioned, each request $q_i$ has a corresponding release time $t_i$ and a location $x_i$. 
We have a set of predictions $P=\{p_1, \dots, p_n \}$ in which $p_i$ predicts $x_i$, the location of request $q_i$.
The algorithm gets these predictions as well as the number of requests $n$ as an offline input. The actual values of $x_i$ and $t_i$ only become known at time-point $t_i$. 

The predictions' quality can vary and is unknown to the algorithm. We can evaluate the quality by defining a measure $\eta$. Essentially, $\eta$ measures the sum of all the distance between the predicted location and the actual ones normalized by the length of a shortest path serving all the requests.

\begin{definition}[Prediction Error]
The prediction error of an instance is defined by $\eta = \frac{\sum_{i=1}^{n} d(x_i, p_i)}{F}$, where $F$ is the length of a shortest path serving all the requests (and returning to the origin in the closed case).
\end{definition}

 
Note that the prediction error is scale invariant (i.e., it will not change if all the distances are multiplied by a constant factor), and $F$ acts as the normalization factor.

    \section{Oracle-based framework: the SWAG algorithm}\label{subsec:perfectpred}

In this section, we define an oracle-based algorithm, $\swag$, designed for the case of perfect predictions. The oracle provides the algorithm with a set of permutations of the requests. We show that if the oracle satisfies some conditions, then $\swag$ has a competitive ratio of 3/2. 

$\swag$ is actually a (slightly) modified version of the general algorithm in~\cite{tsp_l}. The principle of this latter algorithm is the following: 
\begin{itemize}
    \item First, to wait at $\orig$ until a chosen time $T$. This time $T$ depends both on the requests' locations and on their release times.
    \item Then, to choose a route of serving requests that minimizes some criterion involving the length of the corresponding route and the fraction of it which is released at time $T$, and to follow this route, waiting at unreleased requests.
\end{itemize} 
The calculation of the chosen time $T$ on the first step is done by computing some values on {\it all} the $n!$ permutations of the requests. The improvement we show here is that we can get the same competitive ratio (i.e., 3/2) while considering not the whole set of permutations, but some well chosen (and ideally small) subset that satisfies a certain property. This subset of permutations is given to the algorithm by the oracle. 

Then, based on this framework, to derive an efficient 3/2-competitive algorithm, one only has to build an efficient oracle. We will show for example in Section~\ref{sec:polyTimeAlgos} that for lines or rings, we can devise a polytime oracle building a polysize subset of permutations, leading to a polytime 3/2-competitive algorithm for these metrics.

More formally, we consider $\swag$, which uses the following notation.
For a given order $\sigma$ on the requests (where $\sigma[i]$ denotes the $i$-th request in the order), we denote:
\begin{itemize} \item by $\ell_\sigma$ the length of the route associated to $\sigma$ (starting at $\orig$), i.e., $\ell_\sigma=d(\orig,\sigma[1])+\sum_{j=1}^{n-1}d(\sigma[j],\sigma[j+1])$ in the open case, $\ell_\sigma=d(O,\sigma[1])+\sum_{j=1}^{n-1}d(\sigma[j],\sigma[j+1])+d(\sigma[n],\orig)$ in the closed case; 
\item by $\alpha_{\sigma}(t)$ the fraction of the length of the largest fully released prefix of the route associated to $\sigma$ at time $t$ over $\ell_{\sigma}$. More formally, if all $n$ requests are released, then $\alpha_{\sigma}(t) = 1$ for all $\sigma$, otherwise, if requests $\sigma[1],\dots,\sigma[k-1]$ are released at $t$ but $\sigma[k]$ is not, then the route is fully released up to $\sigma[k]$, and $$\alpha_{\sigma}(t)={\big(d(\orig,\sigma[1])+\sum_{j=1}^{k-1} d(\sigma[j],\sigma[j+1])} \big)/{\ell_\sigma} \enspace.$$
Note that this definition requires $\ell_{\sigma} > 0$. If $\ell_{\sigma} = 0$, i.e., all requests are at $\orig$, we set $\alpha_{\sigma}=1$.
\end{itemize}

\begin{algorithm}[ht]
    \KwIn{Offline: request locations $x_1, \dots, x_n$ \\
	\phantom{\textbf{Input:} }Online: release times $t_1, \dots, t_n$ \\
    \phantom{\textbf{Input:} }Parameter: an oracle $\oracle$ which outputs a set of permutations on requests }
    Call the oracle $\oracle$ to get an initial set $S(0)$ of permutations at $t=0$. Set $S = S(0)$.\\
    \While{true}{
        At each release time $t_i$, request a new set $S(t_i)$ of permutations and update $S = S(t_i)$.
        For every $\sigma \in S$, compute $\ell_{\sigma}$ and $\alpha_{\sigma}(t_i)$.\\
        If $\exists \: \firingPerm \in S$ s.t.\ (1) $t\geq \ell_{\firingPerm}/2$ and (2) $\alpha_{\firingPerm}(t)\geq 1/2$, set $\algoStart=t$ and \textbf{break}.
    }
	At time $\algoStart$:
	\begin{itemize}\item Compute an order $\chosenPerm$ which minimizes, over all orders $\sigma' \in S$ , $(1-\beta_{\sigma'})\ell_{\sigma'}$,\\
    where  $\beta_{\sigma'}=\min\{\alpha_{\sigma'}(\algoStart),1/2\}$.

	\item Follow the tour/path associated to $\chosenPerm$. Serve the
    requests in this order,\\ 
    waiting at a request location if this request is not released.
	\end{itemize}
    \caption{Strategically Wait And Go (SWAG)}
	\label{algo:general}
\end{algorithm}

The algorithm for general metrics in~\cite{tsp_l} considers all possible permutations of requests to find the waiting time $\algoStart$ and a good permutation to follow. With $\swag$, we build on this idea. Crucially, at any point in time $t$, it focuses only on a subset $S$ of permutations to determine $\algoStart$ and the permutation to follow. The key idea behind these subsets is the following. 
To achieve the same competitive ratio, it is sufficient that for any possible permutation $\sigma$, the set $S$ contains a permutation $\sigma'$ that induces a tour/path that is not longer than that of $\sigma$, and its unreleased portion of the tour at time $t$ is not larger than that of $\sigma$.
We will say that $\sigma'$ dominates $\sigma$ and define this notion formally below. Furthermore, since the released parts only change when a new request is released, it is sufficient to only update the subset in that case.

\begin{definition}[Dominating permutation]\label{Def:DominationPermutation} Let $\sigma$ be a permutation of the $n$ requests and $t$ a given time. We define Dom($\sigma,t$) to be the set of permutations that dominate permutation $\sigma$ at time $t$. A permutation $\sigma'\in Dom(\sigma,t)$ if and only if:
    \begin{equation} \label{gen:length}
        \ell_{\sigma'}\leq \ell_{\sigma} \enspace,
    \end{equation}
    and
    \begin{equation} \label{gen:unreleased}
        (1-\alpha_{\sigma'}(t))\ell_{\sigma'}\leq (1-\alpha_{\sigma}(t))\ell_\sigma \enspace.
    \end{equation}
    We also say that $\sigma'$ is a corresponding dominating permutation of $\sigma$ (at time $t$).
\end{definition}

For ease of exposition, we assume that the subsets of permutations are given to the algorithm by an oracle $\oracle$. We will discuss in Sections~\ref{sec:singleExponential} and~\ref{sec:polyTimeAlgos} how to implement such oracles for specific cases.

We show that $\swag$ with an oracle $\oracle$ is $3/2$-consistent if $\oracle$ is a \emph{domination oracle}. Formally, we have the following.

\begin{definition}[Domination oracle]\label{def:oracle}
    An oracle $\oracle$ which outputs at time $t$ a set $S(t)$ of permutations is a domination oracle if 
    \begin{enumerate}
        \item $S(t) \subseteq S(t')$ for every $t \le t'$, and
        \item for all $t$ there exists a permutation $\sigma' \in S(t)$ such that $\sigma' \in \text{Dom}(\sigma_{\opt},t)$. Here, $\sigma_{\opt}$ is the permutation corresponding to the serving order of requests in an optimal solution.
    \end{enumerate}
\end{definition}

\begin{lemma} \label{dominating:general}
    $\swag$ is $3/2$-consistent for both closed and open variants of OLTSP with perfect predictions if it uses a domination oracle.
\end{lemma}

\begin{proof}
    We denote by $S(t)$ the set of permutations that the algorithm considers at time $t$ (set $S$ in the description of $\swag$) and by $t_q$ the release time of a request $q$. Note that the few times we drop the dependence on $t$ we imply $t=\algoStart$.

    First, we show that the algorithm always terminates since the while loop in $\swag$ always terminates when a domination oracle is used. When all requests have been released we have that $\alpha_{\sigma}(t) = 1$ for every permutation $\sigma$ of requests. From Equation~\eqref{gen:unreleased} and the domination oracle's condition~2 we get that there is $\sigma' \in S(t)$ such that $(1-\alpha_{\sigma'}(t))\ell_{\sigma'}\leq 0$. Then $\alpha_{\sigma'}(t)=1$ (when $\ell_{\sigma'}=0$ we also have that $\alpha_{\sigma'}(t)=1$ from the definition of $\alpha$). So, there is a time $t'$ such that $t' = \ell_{\sigma'}/2$ and the loop terminates no later than $t=t'$.

    At time $\algoStart$, there exists a permutation $\firingPerm \in S(\algoStart)$ with $\algoStart\geq \ell_{\firingPerm}/2$ and $\alpha_{\firingPerm}(\algoStart)\geq 1/2$. Then $\beta_{\firingPerm}=1/2$, and $(1-\beta_{\firingPerm})\ell_{\firingPerm}=\ell_{\firingPerm}/2\leq \algoStart$. By definition of  $\chosenPerm$, we have $(1-\beta_{\chosenPerm})\ell_{\chosenPerm}\leq (1-\beta_{\firingPerm})\ell_{\firingPerm}$. So we get
    \begin{equation}\label{eq:gen1'}
     (1-\beta_{\chosenPerm})\ell_{\chosenPerm}\leq \algoStart\enspace.
    \end{equation}

     Now let us consider the optimal solution $\opt$ and its corresponding permutation $\sigma_{\opt}$. Consider at time $t$ a corresponding dominating permutation of $\opt$, $\sigma'(t) \in S(t)$ (domination oracle's condition~2). Since $\sigma'$ will always refer to a permutation that dominates $\opt$, we drop the $\opt$ in the notation.
    From Equation~(\ref{gen:unreleased}), it follows  that
    \begin{equation*}
        |\opt| \ge \algoStart + (1-\alpha_{\sigma_{\opt}}(\algoStart))\ell_{\sigma_{\opt}} \ge \algoStart + (1-\alpha_{\sigma'(\algoStart)}(\algoStart))\ell_{\sigma'(\algoStart)}\enspace.
    \end{equation*}
        We distinguish two cases for the value of $\alpha_{\sigma'(\algoStart)}(\algoStart)$.
        
        If $\alpha_{\sigma'(\algoStart)}(\algoStart) \le 1/2$, then  $\beta_{\sigma'(\algoStart)}(\algoStart) = \alpha_{\sigma'(\algoStart)}(\algoStart)$, and 
            \begin{equation*}
                |\opt| \ge \algoStart + (1-\alpha_{\sigma'(\algoStart)}(\algoStart))\ell_{\sigma'(\algoStart)} = \algoStart + (1-\beta_{\sigma'(\algoStart)}(\algoStart))\ell_{\sigma'(\algoStart)}\enspace.
            \end{equation*}
            Since $\chosenPerm$ is a minimizer for $(1-\beta)\ell$ at $\algoStart$ we get
            \begin{equation*}
                   |\opt| \ge \algoStart+(1-\beta_{\chosenPerm})\ell_{\chosenPerm}\enspace.
            \end{equation*}

        Otherwise, if $\alpha_{\sigma'(\algoStart)}(\algoStart) > 1/2$, then  $\beta_{\sigma'(\algoStart)}(\algoStart) = 1/2$. We study two sub-cases.

        \begin{itemize}
            \item \textbf{Case 1:} $\algoStart \neq t_q$ for every request $q$. Since the algorithm updates set $S(t)$ only when there is a release of a new request, the corresponding dominating permutation of $\opt$, $\sigma'(\algoStart)$, is the same at $\algoStart-\epsilon'$, $\forall \epsilon' \in (0,\epsilon_1]$ for a sufficiently small $\epsilon_1 > 0$ (we can choose $\epsilon_1$ such that there is no release of a new request in $[\algoStart-\epsilon',\algoStart]$, $\forall \epsilon' \in (0,\epsilon_1]$). So, we have that $\sigma'(\algoStart)  = \sigma'(\algoStart-\epsilon')$, $\forall \epsilon' \in (0,\epsilon_1]$.

            We now look at the position of $\opt$ at time $\algoStart$ in its route. Suppose that it is (strictly) on the second half of this route. Then, at $\algoStart-\epsilon$ (for a sufficiently small $\epsilon_1 \ge \epsilon>0$), it was already on the second part. But then $\algoStart-\epsilon \ge \ell_{\sigma_{\opt}}/2$ (as $\opt$ has already visited half of the route). So, from Equation~(\ref{gen:length}) we get $\algoStart-\epsilon \ge \ell_{\sigma_{\opt}}/2 \ge \ell_{\sigma'(\algoStart-\epsilon)}/2$. Moreover, since in $[\algoStart-\epsilon, \algoStart]$ there is no release of a new request, we have that 
            $\alpha_{\sigma'(\algoStart-\epsilon)}(\algoStart-\epsilon) = \alpha_{\sigma'(\algoStart-\epsilon)}(\algoStart) > 1/2$.         
            Then $\algoStart-\epsilon$ would satisfy the two conditions for the starting time of the algorithm, a contradiction with the definition of $\algoStart$. Consequently,  at $\algoStart$, $\opt$ is in the first half of its route. So $|\opt|\ge \algoStart+\ell_{\sigma_{\opt}}/2 \ge \algoStart + \ell_{\sigma'(\algoStart)}/2 = \algoStart + (1-\beta_{\sigma'(\algoStart)}(\algoStart))\ell_{\sigma'(\algoStart)}$. Since $\chosenPerm$ is a minimizer for $(1-\beta)\ell$ at $\algoStart$, we get $|\opt| \ge \algoStart+(1-\beta_{\chosenPerm})\ell_{\chosenPerm}$.

            \item \textbf{Case 2:} $\algoStart = t_q$ for some request $q$. We can now find a (sufficiently small) $\epsilon_1 > 0$ such that there is no release of a new request in $[\algoStart-\epsilon_1,\algoStart)$. We can then find another (sufficiently small) $\epsilon_2$ with $0 < \epsilon_2 \le \epsilon_1$ such that $\opt$ is only in the first half of its route in $[\algoStart-\epsilon_2,\algoStart)$ or only on the second half of its route in $[\algoStart-\epsilon_2,\algoStart)$. Let $\epsilon \in (0,\epsilon_2]$. We now look at the position of $\opt$ at time $\algoStart-\epsilon$ in its route.

            \begin{itemize}
                \item If $\opt$ is in the first half at $\algoStart-\epsilon$, then $|\opt| \ge \algoStart -\epsilon+\ell_{\sigma_{\opt}}/2 \ge \algoStart - \epsilon + \ell_{\sigma'(\algoStart)}/2 = \algoStart-\epsilon + (1-\beta_{\sigma'(\algoStart)}(\algoStart))\ell_{\sigma'(\algoStart)}$. The inequality holds for every choice of $\epsilon \in (0,\epsilon_2]$, so taking the limit $\epsilon \to 0$ we get $|\opt|\ge \algoStart + (1-\beta_{\sigma'(\algoStart)}(\algoStart))\ell_{\sigma'(\algoStart)}$. Since $\chosenPerm$ is a minimizer for $(1-\beta)\ell$ at $\algoStart$, we get
                $|\opt| \ge \algoStart+(1-\beta_{\chosenPerm})\ell_{\chosenPerm}$.

                \item If $\opt$ is (strictly) on the second half of its route at $\algoStart-\epsilon$, then consider a corresponding dominating permutation of $\opt$ at $\algoStart-\epsilon$ that is contained in $S(\algoStart-\epsilon)$, $\sigma'(\algoStart-\epsilon)$.
                    
                Suppose that $\alpha_{\sigma'(\algoStart-\epsilon)}(\algoStart-\epsilon) > 1/2$. Then, since $\opt$ is (strictly) on the second half at $\algoStart-\epsilon$, we have that $\algoStart - \epsilon \ge \ell_{\sigma_{\opt}}/2$ (as $\opt$ has already visited half of the route). So, from Equation~\eqref{gen:length} we get $\algoStart-\epsilon \ge \ell_{\sigma_{\opt}}/2 \ge \ell_{\sigma'(\algoStart-\epsilon)}/2$. Then $\algoStart-\epsilon$ would satisfy the two conditions for the starting time of the algorithm, a contradiction with the definition of $\algoStart$. So, if $\opt$ is (strictly) on the second half of its route at $\algoStart-\epsilon$, it holds that $\alpha_{\sigma'(\algoStart-\epsilon)}(\algoStart-\epsilon) \le 1/2$.

                So, $\alpha_{\sigma'(\algoStart-\epsilon)}(\algoStart-\epsilon) \le 1/2$ and $\beta_{\sigma'(\algoStart-\epsilon)}(\algoStart-\epsilon) = \alpha_{\sigma'(\algoStart-\epsilon)}(\algoStart-\epsilon)$. From Equation~\eqref{gen:length} we have that 
                \begin{align*}
                    |\opt| &\ge \algoStart -\epsilon + (1-\alpha_{\sigma_{\opt}}(\algoStart-\epsilon))\ell_{\sigma_{\opt}} \\
                    &\ge \algoStart-\epsilon+(1-\alpha_{\sigma'(\algoStart-\epsilon)}(\algoStart-\epsilon))\ell_{\sigma'(\algoStart-\epsilon)} \\
                    &= \algoStart-\epsilon+(1-\beta_{\sigma'(\algoStart-\epsilon)}(\algoStart-\epsilon))\ell_{\sigma'(\algoStart-\epsilon)}\enspace.
                \end{align*}
                Note that at time $\algoStart$ the set $S(\algoStart)$ also contains the permutation $\sigma'(\algoStart-\epsilon)$ that dominates $\sigma_{\opt}$ at  $\algoStart-\epsilon$ (domination oracle's condition~1). So, $\alpha_{\sigma'(\algoStart-\epsilon)}(\algoStart-\epsilon) \le \alpha_{\sigma'(\algoStart-\epsilon)}(\algoStart)$ and $\beta_{\sigma'(\algoStart-\epsilon)}(\algoStart-\epsilon) \le \beta_{\sigma'(\algoStart-\epsilon)}(\algoStart)$, for every $\epsilon \in (0,\epsilon_2]$ ($\alpha$ and $\beta$ are non decreasing functions of time for a fixed permutation). Thus, we get $                    |\opt| \ge \algoStart-\epsilon+(1-\beta_{\sigma'(\algoStart-\epsilon)}(\algoStart))\ell_{\sigma'(\algoStart-\epsilon)}$.
                This implies $
                    |\opt| \ge \algoStart-\epsilon+(1-\beta_{\chosenPerm})\ell_{\chosenPerm}$ since $\chosenPerm$ is a minimizer for $(1-\beta)\ell$ at $\algoStart$.
                As the above inequality holds for every $\epsilon \in (0,\epsilon_2]$, for $\epsilon \to 0$ it follows that
                $
                    |\opt| \ge \algoStart+(1-\beta_{\chosenPerm})\ell_{\chosenPerm}
                $.
            \end{itemize}
        \end{itemize}

    Consequently, in all cases we have that
    \begin{equation}\label{eq:gen2'}
         |\opt| \ge \algoStart+(1-\beta_{\chosenPerm}(\algoStart))\ell_{\chosenPerm}\enspace.
    \end{equation}

    Now we look at the value $|\alg|$ of the solution $\alg$ output by the algorithm. We distinguish two cases:
    \begin{itemize}
        \item If $\alg$ does not wait after $\algoStart$, it holds that $|\alg|=\algoStart+\ell_{\chosenPerm}$. By definition $\beta_{\chosenPerm}(\algoStart)\leq 1/2$, so
        $|\alg|\leq \algoStart+2(1-\beta_{\chosenPerm}(\algoStart))\ell_{\chosenPerm}$. Adding Equations~(\ref{eq:gen1'}) and~(\ref{eq:gen2'})
        with coefficients $1/2$ and $3/2$, we get
        $\algoStart+2(1-\beta_{\chosenPerm}(\algoStart))\ell_{\chosenPerm}\leq 3|\opt|/2$. 
        Hence, $|\alg|\leq 3|\opt|/2$.  
        \item Otherwise, $\alg$ waits after $\algoStart$ for some request to be released. Let $t^*$ be the last time $\alg$ waits. As a fraction $\alpha_{\chosenPerm}(\algoStart)$ of $\chosenPerm$ is completely released at $\algoStart$ (i.e., when $\alg$ starts), $\alg$ has distance at most  $(1-\alpha_{\chosenPerm}(\algoStart))\ell_{\chosenPerm}$ to perform after $t^*$. So
        \begin{equation}\label{eq:gen3'}
         |\alg|\leq t^*+(1-\alpha_{\chosenPerm}(\algoStart))\ell_{\chosenPerm} \leq t^*+(1-\beta_{\chosenPerm}(\algoStart))\ell_{\chosenPerm}\enspace,
        \end{equation}
        where we use the fact that, by definition,  $\beta_{\chosenPerm}(\algoStart)\leq \alpha_{\chosenPerm}(\algoStart)$.
        We have $t^*\leq |\opt|$, as a request is released at $t^*$. Adding Equations~(\ref{eq:gen1'}) and~(\ref{eq:gen2'}) gives $2(1-\beta_{\chosenPerm}(\algoStart))\ell_{\chosenPerm}\leq |\opt|.$ Putting these two inequalities in Equation~(\ref{eq:gen3'}) gives $|\alg|\leq 3|\opt|/2$.
    \end{itemize}
\end{proof}

Let us consider now the running time of $\swag$ and prove the following lemma. 

\begin{lemma} \label{general:time}
 If $N$ is the maximum number of permutations which the oracle $\oracle$ outputs at each time $t$ and $T_\oracle$ is the total time required (by the oracle) to compute all permutations, then the running time of $\swag$ is $O(\max\{n^2\cdot N, T_\oracle\})$.
\end{lemma}
\begin{proof}
    For every release time $t_i$ and every permutation $\sigma \in S(t_i)$, the algorithm has to compute the length $\ell_{\sigma}$ and the fraction of the associated tour that is released, $\alpha_{\sigma}$. It can do that with $O(n)$ operations and has to repeat this procedure $n+1$ times, once for every release of a new request (and once at $t=0$). If $N$ is the maximum number of permutations which the oracle $\oracle$ outputs, then the total running time of the algorithm without taking the oracle's workload into account is $O(n^2\cdot N)$. Moreover, if we assume that for all calls of the oracle $\oracle$ the total computational time required by $\oracle$ to compute all permutations is $T_\oracle$, then the total running time of $\swag$ is $O(\max\{n^2\cdot N, T_\oracle\})$. 
\end{proof}

    \section{Performance guarantees for LA-OLTSP: the LA-SWAG algorithm}\label{subsec:smooth}

In this section, we deal with imperfect predictions; more specifically, we show how to adapt $\swag$ to also get smoothness and robustness upper bounds. We first present the algorithm, then the consistency-smoothness analysis, and finally the robustness analysis.

\begin{algorithm}[ht]
	\KwIn{Offline: predicted request locations $p_1, \dots, p_n$\\ 
	\phantom{\textbf{Input:} }Online: release times $t_1, \dots, t_n$, true request locations $x_1,\dots,x_n$ \\
    \phantom{\textbf{Input:} }Parameter: an oracle $\oracle$ which outputs a set of permutations on requests} 

    \medskip
    
    (Breaking rule) At any time $t$: if all requests are released, follow an optimal path serving all unserved requests (returning to $\orig$ if in the closed variant) and \textbf{break}.\\

    \medskip
    
    Run $\swag$ until the starting time $\algoStart$ of the server, and the computation of $\chosenPerm$.\\ 

    \medskip
    
    At time $\algoStart$, follow the tour/path $\chosenPerm$, serving the requests in the following order:\\
    \For{$i=1,\dots,n$}{
    \begin{itemize}
        \item first go to $p_{\chosenPerm[i]}$; if  $q_{\chosenPerm[i]}$ is not released, wait there until it is released;
        \item then, go to $x_{\chosenPerm[i]}$ and serve the request.
    \end{itemize}}
     In the closed version, go back to $\orig$. 

    \caption{Learning-Augmented Strategically Wait And Go (LA-SWAG) }
	\label{algo:generalSmoothRobust}
\end{algorithm}

We note that if the predictions are perfect, then $\laswag$ is at least as good as $\swag$ (it works the same, the only difference being that it optimally serves the remaining unserved requests when everything is released). So in particular it is 3/2-consistent as well, provided that the oracle satisfies the conditions of Lemma~\ref{dominating:general}. 

We also note that the algorithm could serve the request in a more clever way: instead of going first to the predicted location of a request and then to its true location, the algorithm could go to the true location directly if it is released (or as soon as it is). However, the lower bounds we will provide still hold under this modification. 

Regarding the running time of $\laswag$, it is trivial to see that if $T_{TSP}$ is the time required for the computation of an optimal path serving a subset of the requests, then from Lemma~\ref{general:time} we get the following.

\begin{corollary} \label{general:time-2}
 If $N$ is the maximum number of permutations which the oracle $\oracle$ outputs at each time $t$, $T_\oracle$ is the total time required (by the oracle) to compute all permutations and $T_{TSP}$ is the time for the computation of an optimal path that serves a subset of the requests, then the running time of $\laswag$ is $O(\max\{n^2\cdot N, T_\oracle, T_{TSP}\})$.
\end{corollary}

Note that, for the cases we consider, we achieve suitable bounds on $T_{TSP}$. We describe how this is this done for each metric space in Section \ref{sec:polyTimeAlgos}.

\subsection{Smoothness}

We now show that $\laswag$ is smooth with respect to the measure of error $\eta$ in the predictions, where we recall that
$\eta = {\sum_{i=1}^{n} d(x_i, p_i)}/{F}$ with $F$ the length of a shortest TSP-tour serving all the requests (and returning to the origin in the closed case), ignoring release times.
For convenience, let $\Delta = \eta F = \sum_{i=1}^{n} d(x_i, p_i)$.

\ConsistencySmoothness*

\begin{proof}
Let $I$ be an instance of the problem with requests $q_1, \dots, q_n$, predictions $p_1,\dots, p_n$ and actual locations $x_1,\dots,x_n$.   
Let us denote by $I_H$ a corresponding (hypothetical) instance of the problem where the predicted locations $p_1,\dots,p_n$ would be the actual ones (so $I_H$ is an instance with perfect predictions). Let us denote by $\alg_{I_H}$ (resp. $\opt_{I_H}$) the solution output by $\swag$ on $I_H$ (resp. an optimal solution on $I_H$). By Lemma~\ref{dominating:general}, we have $|\alg_{I_H}|_{I_H} \le 3/2|\opt_{I_H}|_{I_H}$ since in $I_H$ the predictions are perfect. 

Now, we first claim that $|\alg|_I \le |\alg_{I_H}|_{I_H} + 2\Delta$. As previously noticed, the breaking rule in $\laswag$ cannot increase the value of the solution, so $|\alg|_I\leq |\alg'|_I$ where $\alg'$ is the solution on $I$ that would perform the entire for loop in Step~4 of $\laswag$.  Note that now $\alg_{I_H}$ and $\alg'$ start at the same time $\algoStart$, and serve the requests in the same order $\chosenPerm$. 
Let $t^*$ be the last time $\alg'$ waits on a prediction $p_{\chosenPerm[s]}$ for the actual instance $I$, or $t^*=\algoStart$ if $\alg'$ never waits on a prediction. In the first case,  $\alg_{I_H}$ cannot serve the corresponding request $\chosenPerm[s]$ before $t^*$ (as it is not released). In the second case, $\alg_{I_H}$ does not leave the origin before $t^*$. Thus, denoting for convenience $s=0$, ${\chosenPerm[0]}=0$ and $p_{0}=\orig$ if $\alg'$ never waits at a prediction, we  have for the closed case (just remove the term $d(p_n, \orig)$ for the open case):
\begin{displaymath}
|\alg_{I_H}|_{I_H} \ge t^* + \sum_{i = s}^{n - 1} d(p_{\chosenPerm[i]}, p_{\chosenPerm[i+1]}) + d(p_{\chosenPerm[n]}, \orig) \enspace.
\end{displaymath}
Since the hypothetical algorithm $\alg'$ moves between predictions and requests, we get the following bound on $|\alg'|_I$
\begin{displaymath}
\begin{gathered}
    |\alg'|_I \le t^* + \left(\sum_{i = s}^{n - 1} d(p_{\chosenPerm[i]}, x_{\chosenPerm[i]}) + d(x_{\chosenPerm[i]}, p_{\chosenPerm[i+1]})\right) + d(p_{\chosenPerm[n]}, x_{\chosenPerm[n]}) + d(x_{\chosenPerm[n]}, \orig) \enspace.
\end{gathered}
\end{displaymath}
Using the triangle inequality $d(x_{\chosenPerm[i]}, p_{\chosenPerm[i+1]}) \le d(x_{\chosenPerm[i]}, p_{\chosenPerm[i]}) + d(p_{\chosenPerm[i]}, p_{\chosenPerm[i+1]})$, we get 
\begin{align*}
    |\alg'|_I &\le 
    t^* + \left(\sum_{i = s}^{n - 1} d(p_{\chosenPerm[i]}, p_{\chosenPerm[i+1]}) + 2d(x_{\chosenPerm[i]}, p_{\chosenPerm[i]})\right) + 2d(p_{\chosenPerm[n]}, x_{\chosenPerm[n]}) + d(p_{\chosenPerm[n]}, \orig) \\
    &\le |\alg_{I_H}|_{I_H} + 2\Delta \enspace.
\end{align*}

Now, we shift our attention to the values $|\opt_I|_I$ and $|\opt_{I_H}|_{I_H}$. Our goal will be to show that $|\opt_I|_I \ge |\opt_{I_H}|_{I_H} - 2\Delta$. Since $|\opt_{I_H}|_{I_H} \le |\opt_I|_{I_H}$, it suffices to show $|\opt_I|_I \ge |\opt_I|_{I_H} - 2\Delta$. We can assume that $\opt$ moves to the requests eagerly, following an order $\sigma$, and leaves as soon as they are released. Similarly to before, let $t^*$ be the last time that the tour of $\opt_I$ in the instance $I_H$ induces some waiting, and let $q_{\sigma[s]}$ be that request (or the origin for $s = 0$, which also implies $t^* = 0$). It follows then that 

\begin{displaymath}
    |\opt_I|_{I_H} = t^* + \sum_{i = s}^{n - 1} d(p_{\sigma[i]}, p_{\sigma[i+1]}) + d(p_{\sigma[n]}, \orig)\enspace.
\end{displaymath}

But, because the request $q_{\sigma[s]}$ is unreleased before $t^*$ (the case $s = 0$ implies $t^* = 0$, so the conclusion will still hold), we also have
\begin{align*}
    |\opt_I|_I &\ge t^* + \sum_{i = s}^{n - 1} d(x_{\sigma[i]}, x_{\sigma[i+1]}) + d(x_{\sigma[n]}, \orig) \\
    &\ge t^* + \left(\sum_{i = s}^{n - 1} -d(x_{\sigma[i]}, p_{\sigma[i]}) + d(p_{\sigma[i]}, p_{\sigma[i+1]}) - d(p_{\sigma[i+1]}, x_{\sigma[i+1]})\right) \\*
    &\hspace{22pt}- d(x_{\sigma[n]}, p_{\sigma[n]}) + d(p_{\sigma[n]}, \orig) \\
    &\ge
    |\opt_I|_{I_H} - 2\Delta\enspace,
\end{align*}
where we use the triangle inequality for the second inequality.
The smoothness of the general algorithm follows. We have $|\alg_{I_H}|_{I_H} \le 3/2|\opt_{I_H}|_{I_H}$, which implies $|\alg|_I - 2\Delta \le
    3/2(|\opt_I|_I + 2\Delta)$. This gives us $\frac{|\alg|_I}{|\opt_I|_I} \le 3/2 + 5 \eta.$
\end{proof}

$\laswag$ achieves smoothness that is linear in the error (besides the optimal consistency bound of 3/2). We now show that this linear dependency is necessary for the open case.

\begin{proposition}[Smoothness Lower Bound]\label{smoothnessLowerBound}
No algorithm can have a better competitive ratio than $\left(3/2+\eta/2\right)$ for the open variant on an instance with prediction error $\eta\in [0,1/3]$.
\end{proposition}
\begin{proof}
    Let $\eta\in [0,1/3]$, and define $\epsilon$ such that $\eta=\epsilon/(2-\epsilon)$ (i.e., $\epsilon=2\eta/(1+\eta)$). Note that $\epsilon \in [0,1/2]$. Consider the graph represented in Figure~\ref{fig:smoothLB}. 
    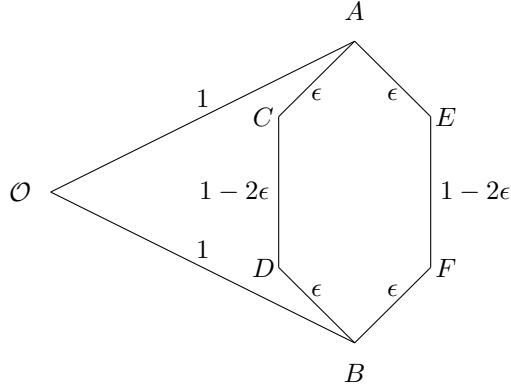
\begin{figure}[ht]
    \begin{center}
\begin{tikzpicture}[scale=2]
	\draw (0,0)-- (2,1) node[midway,above] {$1$};
	\draw (0,0) -- (2,-1)  node[midway,above] {$1$};
	\draw (2,1) -- (1.5,0.5)  node[midway,below] {$\epsilon$};
 	\draw (1.5,0.5) -- (1.5,-0.5)  node[midway,left] {$1-2\epsilon$};
   	\draw (1.5,-0.5) -- (2,-1)  node[midway,above] {$\epsilon$};
    	\draw (2,1) -- (2.5,0.5)  node[midway,below] {$\epsilon$};
 	\draw (2.5,0.5) -- (2.5,-0.5)  node[midway,right] {$1-2\epsilon$};
   	\draw (2.5,-0.5) -- (2,-1)  node[midway,above] {$\epsilon$};
	\node at (-0.2,0.0) (origin) {$\orig$};
 	\node at (2,1.2) (origin) {$A$};
   	\node at (2,-1.2) (origin) {$B$};
  \node at (1.4,0.5) (origin) {$C$};
  \node at (1.4,-0.5) (origin) {$D$};
  \node at (2.6,0.5) (origin) {$E$};
  \node at (2.6,-0.5) (origin) {$F$};  
	\end{tikzpicture}
\end{center}
\caption{An example for the open variant showing a metric space in which the competitive ratio has to scale linearly with $\eta$ for values of the latter up to $1/3$. }\label{fig:smoothLB}
\end{figure}
    
    There are two requests, $q_1$ predicted to be in location $p_1=A$ and $q_2$ predicted to be in $p_2=B$. We assume w.l.o.g.\ that at time $t=1$ the algorithm is on $[O,B]$. Then $q_1$ is released at $t=1$ on its predicted location ($x_1=p_1=A$). Note that at time  $2-\epsilon$, the algorithm must be at distance at least $\epsilon$ from $A$ (in particular, it cannot be on $(C,A]$ or on $(E,A]$). If it is on $[B,E]$, then we release the second request $q_2$ on $D$. Otherwise, we release it on $F$. Suppose w.l.o.g.\ that the first case occurs. Then the algorithm needs at least one unit of time to serve both requests (at $t = 2-\epsilon$ it is not on $(B,C)$ and at distance at least $\epsilon$ from $A$). So $|\alg|\geq 3-\epsilon$. 

    The optimal solution is to serve $q_1$ at $t=1$ and $q_2$ at $t=2-\epsilon.$ Note that $\eta=\frac{\epsilon}{2-\epsilon}$ is indeed the error in the prediction. Then $$|\alg|\geq \frac{3-\epsilon}{2-\epsilon} |\opt|=\left(\frac{3}{2}+\frac{\eta}{2}\right)|\opt|\enspace.$$
\end{proof}

    \subsection{Robustness}

We now deal with the robustness analysis. We first upper bound the robustness of $\laswag$, and then complement this analysis with some lower bounds.

\begin{remark}[3-Robustness]\label{3-robust}
It is easy to see that $\laswag$ is 3-robust (in both the closed and open variants). Indeed, let $t_f$ be the time when the final request is released. Of course $|\opt|\geq t_f$. Then, at $t_f$ the server can be at most $|\opt|$ away from the origin. Since, thanks to the breaking rule, the algorithm serves the remaining requests optimally after $t=t_f$, to return back to the origin and then follow $\opt$ is an upper bound of its cost. So, $|\alg| \le t_f+|\opt|+|\opt| \le 3|\opt|$ and the robustness of the algorithm with predictions is at most $3$.
\end{remark}

We now improve this bound, both for the closed and open cases, both in general metrics and for specific metric spaces. 

Let us first deal with the closed variant.




\paragraph*{Closed variant.}
\begin{lemma} \label{robustness:2.75a}
$\laswag$ is $2.75$-robust for closed OLTSP in general metric spaces.
\end{lemma}

\begin{proof}
Let $\algoStart \ge 0$ be the time computed in Step 2 of  $\laswag$. It depends both on the predicted locations of the requests and the (real) release times. Let $c \in (0,1]$ be a real constant to be determined later. We denote by $\rho$ the ratio $|\alg|/|\opt|$.
\begin{itemize}
    \item If $\algoStart > |\opt|$, then $\alg$ starts following an optimal tour from the origin at time $t=t_f \le |\opt|$ because of the breaking rule. So, $|\alg| \le 2|\opt|$ and 
    \begin{equation} \label{robustness:rho_1}
        \ratio \le 2 \enspace.
    \end{equation}
    
    \item If $0 \le \algoStart \le c \cdot |\opt|$, then there is an order of requests $\firingPerm$ that the general algorithm keeps track of such that $\algoStart \ge \ell'_{\firingPerm}/2$, $\alpha'_{\firingPerm}(\algoStart) \ge 1/2$ and $\beta'_{\firingPerm}(\algoStart)=1/2$, where the parameters with the prime symbol $\ell',\alpha', \beta'$ are computed on the predicted locations and the real release times.

    Let $\chosenPerm$ be the tour that our algorithm will follow at $\algoStart$. Thus, we have that $$\big(1-\beta'_{\chosenPerm}(\algoStart)\big)\ell'_{\chosenPerm} \le \big(1-\beta'_{\chosenPerm}(\algoStart)\big)\ell'_{\chosenPerm} = \ell'_{\chosenPerm}/2 \le \algoStart \enspace.$$
    Since $\beta' \le 1/2$, it holds that ${1-\beta'_{\chosenPerm}(\algoStart)} \ge 1/2$ and $\ell'_{\chosenPerm} \le 2\algoStart$. Hence, with the assumption $\algoStart \le c\cdot |\opt|$ we get that
    \begin{equation} \label{robust_c}
         \ell'_{\chosenPerm} \le 2c \cdot |\opt| \enspace.
    \end{equation}

    After $\algoStart$, the algorithm follows the order of requests $\sigma_1$. It first goes to the predicted location of a request and waits for its release. Then, it goes to the real location of the request and serves it. When all requests are released, the algorithm computes and follows an optimal path that serves the remaining requests and returns to the origin.
    
    The furthest distance from the origin that a predicted request can have is $L_P \le \ell'_{\chosenPerm}/2$. From~\eqref{robust_c}, we get $L_P \le c \cdot |\opt|$. On the other hand, the furthest distance from the origin that a real request can have is $L_R \le \ell_{\opt}/2 \le |\opt|/2$, where $\ell_{\opt}$ is the length of the optimal tour.

    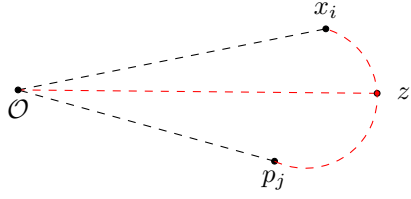
\begin{figure}[ht]
    \begin{center}
    \begin{tikzpicture}[scale=0.9,rotate=270]

		\coordinate (A) at (1,0);
		\coordinate (B) at (0,5);
		\coordinate (O) at ($(A)!0.05!(B)$);
		\coordinate (R) at ($(A)!0.95!(B)$);
		\coordinate (P) at (2,4);
		\coordinate (A1) at (1,5.5);

		\draw[dashed] (O) -- (R);
		\draw[dashed] (P) -- (O);
		\draw[dashed, red] (A1) -- (O);
  
		\draw[fill=black, draw=black] (O) circle (1.2pt);
		\draw[fill=black, draw=black] (R) circle (1.2pt);
		\draw[fill=black, draw=black] (P) circle (1.2pt);
		\draw[fill=red, draw=black] (A1) circle (1.2pt);

		\node[shift={(0,-0.25)}] at (P) {$p_j$};
		\node[shift={(0,-0.25)}] at (O) {$\orig$};
		\node[shift={(0,0.25)}] at (R) {$x_i$};
		\node[shift={(0.35,0)}] at (A1) {$z$};

        \arcThroughThreePoints{P}{A1}{R};

	\end{tikzpicture}
	\end{center}

    \caption{The case for general metrics when the server goes from the furthest predicted location $p_j$ to the location $x_i$ of the furthest (real) request $q_i$ or vice-versa. The algorithm ($z$) is always on the shortest path between $x_i$ and $p_j$.}
    \label{fig:robustness}
    \end{figure}
    
    Now let us consider the furthest distance that $\alg$ can have from the origin at time $t=t_f$. $\alg$ during its tour moves between predicted and real locations of requests until time $t_f$. In the following, we denote the position of $\alg$ at $t_f$ by $z$. The worst-case occurs if $\alg$ goes from the furthest predicted request to the furthest real request or vice-versa (see Fig.~\ref{fig:robustness}). From the triangle inequality, we have that $d(z,\orig) \le d(\orig,x_i)+d(x_i,z)$ and $d(z,\orig) \le d(\orig,p_j)+d(p_j,z)$. Hence, we have $2d(z,\orig) \le d(\orig,x_i)+d(\orig,p_j)+d(x_i,z)+d(p_j,z)$.
     Since the algorithm is on the shortest path between $x_i$ and $p_j$, we get $d(x_i,p_j)=d(x_i,z)+d(p_j,z)$. Using the triangle inequality again it follows that $d(x_i,p_j) = d(x_i,z)+d(p_j,z) \le d(\orig,x_i)+d(\orig,p_j)$.
     Therefore,
     \begin{equation*}
                2d(z,\orig) \le d(\orig,x_i)+d(\orig,p_j)+d(x_i,z)+d(p_j,z) \le 2(d(\orig,x_i)+d(\orig,p_j))\enspace,
     \end{equation*}
     which implies
     \begin{equation} \label{robustness:distance}
            d(z,\orig) \le d(\orig,x_i)+d(\orig,p_j) \le L_R + L_P \le (1/2+c)\cdot |\opt| \enspace.
     \end{equation}
     The strategy that returns to the origin at time $t=t_f$ and then follows $\opt$ is clearly an upper bound of $|\alg|$, and thus from~\eqref{robustness:distance} and the fact that $t_f \le |\opt|$ it follows that $
         |\alg| \le t_f + d(z,\orig) + |\opt| \le (2.5+c)\cdot |\opt| $
     and
     \begin{equation} \label{robustness:rho_2}
         \rho \le 2.5+c \enspace.
     \end{equation}          

    \item Else, $c \cdot |\opt| < \algoStart \le |\opt|.$ Then, $\alg$ is at most $(|\opt|-\algoStart)$ away from the origin at $t=|\opt|$. So, the following strategy is clearly an upper bound of $|\alg|$: an algorithm that is in the origin at $t=|\opt|+|\opt|-\algoStart$ and then follows $\opt$ to serve all requests. Thus, we get
    \begin{equation} \label{robustness:rho_3}
        \ratio \le \frac{|\opt|+|\opt|-\algoStart+|\opt|}{|\opt|} < \frac{(3-c)|\opt|}{|\opt|}=3-c \enspace.
    \end{equation}
\end{itemize}

We can see that the case leading to \eqref{robustness:rho_1} will not be a factor in the determination of the robustness, so we can ignore it. We optimize \eqref{robustness:rho_2} and \eqref{robustness:rho_3} simultaneously. It is easy to see that the best value is attained for $c=0.25$, giving a robustness of $2.75$.
\end{proof}

Moreover, we observe that we can get a better robustness guarantee with almost the same proof for the special metric spaces of Euclidean spaces and trees. More specifically, we get the following lemma.

\begin{lemma}\label{th:2.5rob}
    $\laswag$ is $2.5$-robust for closed OLTSP in Euclidean spaces and in trees. 
\end{lemma}

\begin{proof}
    The proof is almost the same as the proof of Lemma~\ref{robustness:2.75a}. The only difference is that since it holds that $d(z,\orig) \le \max\{d(\orig,x_i), d(\orig,p_j)\} \le \max\{ c\cdot|\opt|, 1/2\cdot|\opt|\}$ in these metric spaces (as the algorithm is on a shortest path between $x_i$ and $p_j$), we can actually get a better upper bound than that of inequality~\eqref{robustness:rho_2}, namely
    \begin{equation} \label{robustness:rho_4}
        \rho \le \frac{t_f+ \max\{d(\orig,x_i), d(\orig,p_j)\}+|\opt|}{|\opt|} \le 2 + \max\{ 1/2, c\} \enspace.
    \end{equation}
    We can again ignore \eqref{robustness:rho_1}. Considering \eqref{robustness:rho_3} and \eqref{robustness:rho_4}, we get robustness $2.5$ with $c=0.5$.
\end{proof}

From Lemma~\ref{robustness:2.75a} and~\ref{th:2.5rob} we directly get the following theorem for the robustness of $\laswag$ in the closed variant.

\robustnessColsed*

\begin{remark}[Tightness of Analysis for Closed Variant 2.5-Robustness.]\label{tight2.5}
The 2.5-robustness bound of $\laswag$ given in Lemma~\ref{th:2.5rob} is tight even on the line, as shown in the following example. There are two requests. The real request locations are $x_1 = 1, x_2 = 0$, released at times 1 and 2, respectively, and the predictions are $p_1 = 0, p_2 = -1$. Note that at time 1 our algorithm will start moving towards the prediction $p_2$, because half the tour $\orig \rightarrow p_2 \rightarrow p_1 \rightarrow \orig$ is ``released'' and enough time (half of its length) has passed.\footnote{In fact, this is the only permutation that $\laswag$ will be provided with by the oracles we define later. To ensure this, we can perturb $p_1$ to the left of the origin by some arbitrarily small $\epsilon$ and then let $\epsilon \to 0$. Hence, it will definitely choose to follow this permutation.} It will keep moving and reach position $-1$ at time $2$, since $p_2$'s corresponding request is unreleased until that time. Then, after reaching this unfavorable position, $\alg$ can only go back to the origin and copy $\opt$. We see that $|\alg| = 2 + |-1| + 2 = 5$ while $\opt = 2$, so the algorithm has competitive ratio $2.5$. 
\end{remark}


\paragraph*{Open variant.}
In fact, we can also show a better than 3 robustness bound even for the open variant, where only the triangle inequality is assumed. The trick is to consider another option for finishing up, namely to go to the last request served by $\opt$, say $q_\ell$, and then copy $\opt$ \textit{backwards}. It is clear that this copy also takes at most $|\opt|$ time, since there might not be a request at the origin. Hence, it remains to show that we can reach $q_\ell$ \textit{or} the origin relatively fast.

\begin{lemma}\label{th:3-1/6}
    $\laswag$ is $\left(3 - \frac{1}{6} \right)$-robust for open OLTSP in general metric spaces.
\end{lemma}

 \begin{proof}
     
     Note that $\laswag$, at the final release time $t_f$, is on \textit{some} shortest path from a prediction to the location of a request or vice versa (or from the origin to a prediction, which is trivial or can be thought of as a request location). Now, let $p$ be that prediction and $q$ the request with actual location $x$. We do not distinguish whether the movement is made from $p$ to $x$ or vice versa. Consider the path $P_1$, which moves from $z$ (the position of the server at time $t_f$) to $p$ and then to $\orig$, and the path $P_2$ which moves from $z$ to $x$ and then to $x_\ell$, the location of the last request served by $\opt$. Now, let $d_p = d(z, p)$ and $d_x = d(z, x)$. Since $\alg$ is on the shortest path from $p$ to $x$ or vice versa, it holds that $d_p + d_x = d(p, x)$. Note that we have
     \begin{align*}
        min\{|P_1|, |P_2|\} &\le \frac{|P_1| + |P_2|}{2} = \frac{d_p + d_x + d(p, \orig) + d(x, x_\ell)}{2} \\
        &\le  d(p, \orig) + \frac{d(x, x_\ell) + d(x, \orig)}{2} \le \frac{|\opt|}{2} + d(p, \orig)\enspace.
     \end{align*}
     This means that we can finish in at most $(2.5 + \gamma)|\opt|$ overall, where $\gamma = \frac{d(p, \orig)}{|\opt|}$. Now, we can apply the same idea from the proofs for the closed case. That is, let $\algoStart = c \cdot |\opt|$, where $\algoStart$ is the start time of the algorithm. We see that $d(p, \orig) \le 2c \cdot |\opt|$, otherwise the conditions necessary for the algorithm to start would not be satisfied at $\algoStart$. Also, a trivial bound for the performance is $(3 - c)|\opt|$. Hence, on the one hand we have $\frac{|\alg|}{|\opt|} \le 2.5 + 2c$ and on the other $\frac{|\alg|}{|\opt|} \le 3 - c$, and these bounds hold simultaneously. Thus, overall we have that $\frac{|\alg|}{|\opt|} \le 3 - \frac{1}{6}$ by setting $c = \frac{1}{6}$.
     \end{proof}

Also, we can improve the bound when dealing with trees.

\begin{lemma}\label{th:3-1/3}
    $\laswag$ is $\left(3 - 1/3\right)$-robust for open OLTSP in trees.
\end{lemma}
\begin{proof}
    Consider a prediction $p$ and the location $x$ of a request $q$.
    For simplicity, let us assume that $p$ and $x$ are vertices in the tree, otherwise we could simply add vertices at these locations. Note that in this case, when moving from $p$ to $x$ or vice versa, $\alg$ first moves to the lowest common ancestor of $p$ and $x$ and then continues along the path to its destination. This means that at any point through this movement, it is on one of $x$ or $p$'s paths to the origin.  In particular, this is also true at $t = t_f$. If it is on the path to $p$, then the distance to the origin is at most $p$'s distance to the origin. If it is on the path to $x$, then it means that it is also on some point along the path that $\opt$ would take. Thus, it can complete the objective by first moving to either endpoint of the path and then copying it, which takes at most $(1 + 2c)|\opt|$ or $3/2|\opt|$ time in total, depending on the previous case distinction, where $c$ is the fraction of $|\opt|$ waited. Hence, we have the bounds $2.5, 2 + 2c, 3 - c$, and so we get the desired competitive ratio by considering $c = 1/3$. 
\end{proof}

The two previous lemmata give us our main robustness theorem for the open variant.

\robustnessOpen*

\begin{remark}[Tightness of Analysis for Open Variant $\left(3 - \frac{1}{3}\right)$-Robustness.]\label{tight3-1/3}
The $\left(3 - \frac{1}{3}\right)$-robustness bound of $\laswag$ given in Lemma~\ref{th:3-1/3} is tight, even on a line. To see this, we place (on the real line) a prediction at $-1$ and the real request at $1.5$, released at $t = 1.5$. Note that at time $t = 0.5$, our algorithm embarks on its path to the prediction. It reaches $-1$ at $t = 1.5$. This means that it cannot finish before $t = 4$, since it has to travel from $-1$ to $1.5$, starting at $t = 1.5$. However, the optimal time is $1.5$, yielding a competitive ratio of $8/3 = 3 - 1/3$.
\end{remark}

\paragraph*{Consistency-Robustness tradeoff lower bound for the open variant.}

We conclude this section by showing a consistency-robustness tradeoff lower bound for the open variant.

\begin{proposition}\label{tradeoff}
    Let $\alg$ be an algorithm for the open variant with consistency guarantee $2 - \lambda$, where $\lambda \in [0, 1]$. Then, $\alg$ cannot be $\left(2 + \lambda - \epsilon\right)$-robust, for any $\epsilon > 0$. 
\end{proposition}

\begin{proof}
    We show this on the line. We place a prediction at $-1$ and the real request at $1$, which will be released at $t=1$. Since $\alg$ is $\left(2 - \lambda\right)$-consistent, it must be located at a position not strictly to the right of $-\lambda$ at $t = 1$, otherwise it could not reach the prediction by time $2 - \lambda$, which would be a contradiction to $\alg$'s consistency if the prediction was correct and released at $t = 1$. Therefore, $\alg$ can only reach the real request at or later than $t = 2 + \lambda$, which concludes the proof since $|\opt| = 1$.
\end{proof}

Note that the above implies that the best possible robustness for any $1.5$-consistent algorithm is $2.5$, placing our algorithm relatively close to the underlying Pareto front.

    \section{A single-exponential time algorithm for general metrics} \label{sec:singleExponential}

In the case of perfect predictions, using $\swag$ and Lemma~\ref{dominating:general} we can reduce the factorial running time of the algorithm in~\cite{tsp_l} for general metrics. We will show that an exponential number of permutations suffices to have a domination oracle. After that, using the scheme of $\laswag$ and the theorems proved in the previous section we will get an exponential-time algorithm that is $3/2$-consistent, smooth and robust.

At each time $t$ there is a set $R(t)$ of released requests and a set $U(t)$ of unreleased requests. For ease of explanation we refer to the closed case here, but the open case is similar. Consider an arbitrary permutation $\sigma$ of the $n$ requests at time $t$. A corresponding dominating tour for this permutation can be constructed as follows. First, find the first unreleased request $u \in U(t)$ in the permutation $\sigma$. The released requests before $u$ in the permutation form a subset $R'(t) \subseteq R(t)$ of released requests. Then, a corresponding general dominating tour would be the optimal TSP path that starts from the origin, serves all requests in $R'(t)$ and ends at request $u$ plus the optimal TSP path that starts from $u$, serves all requests in $\big(R(t) \setminus R'(t)\big) \cup \big(U(t) \setminus \{u\}\big)$ and returns to the origin. We will call these tours (corresponding) \emph{general dominating tours}.

Throughout this section, we will assume for ease of exposition that $t_1 \le t_2 \le \dots \le t_n$ and use $t_0 = 0$ and $t_{n+1}=\infty$. At time $t_i$, with $i \le n$, we compute the set of general dominating tours $D_i$ given the sets $R(t_i)$ and $U(t_i)$. Together with the previously computed sets, it is then provided to $\swag$ as set $S(t_i) = \bigcup_{j=0}^{i} D_j$.

\begin{definition}[General Dominating Permutations]\label{def:GenDomPer}
Let $\sigma$ be a permutation of the $n$ requests. Given a time $t$, there is a set $R(t)$ of released requests and a set $U(t)$ (assume for now that $U(t) \neq \emptyset$) of unreleased requests. Let $u \in U(t)$ be the first unreleased request in the permutation $\sigma$ and $R'(t) \subseteq R(t)$ be the set of released requests that precedes $u$ in the permutation. A corresponding general dominating tour of $\sigma$ at time $t$ consists of:
\begin{enumerate}
    \item an optimal TSP path $TSP_{\orig-R'-u}$ that starts from the origin, serves all requests in $R'(t)$ and ends at request $u$,

    \item and an optimal TSP path $TSP_{u-B-\orig}$ that starts from request $u$, serves all requests in $B=\big(R(t) \setminus R'(t)\big) \cup \big(U(t) \setminus \{u\}\big)$ and returns to the origin.
\end{enumerate}
Note that an optimal TSP path here is computed ignoring the release times of the requests. The length $\ell$ of the corresponding general dominating tour is then
\begin{equation*}
    \ell =  \ell_{TSP_{\orig-R'-u}} + \ell_{TSP_{u-B-\orig}}\enspace.
\end{equation*}

If $U(t) = \emptyset$, then $|R(t)|=n$, every request is released and we define the corresponding general dominating tour for every $\sigma$ to be an optimal TSP tour that starts from the origin, serves all requests and returns again to the origin (the same optimal TSP tour for every $\sigma$).
\end{definition}

The definition of the corresponding general dominating path of a permutation $\sigma$ for the open case is similar, i.e., instead of ending at the origin, a path ends at an arbitrary request from $\big(R \setminus R'\big) \cup \big(U \setminus \{u\}\big)$, and thus we omit it. From now on we will refer to general dominating permutations (paths or tours).

\begin{proposition}
     Let $\sigma$ be a permutation of the $n$ requests. Then, at each time $t$, there exists a corresponding general dominating permutation $\sigma'$ of $\sigma$.  
     Equivalently, it holds that 
     \begin{equation} \label{gen:length_1}
          \ell_{\sigma'} \le \ell_{\sigma}
     \end{equation}
     and
     \begin{equation} \label{gen:unreleased_1}
         (1-\alpha_{\sigma'}(t))\ell_{\sigma'} \le (1-\alpha_{\sigma}(t))\ell_{\sigma} \enspace.
     \end{equation}
 \end{proposition}

 \begin{proof}
     Since time $t$ is fixed here, the time dependence will be suppressed for convenience. At time $t$, there is a set of released requests $R$ and a set of unreleased requests $U$ (suppose for now that $U \ne \emptyset$). Let $u \in U$ be the first unreleased request in the permutation $\sigma$ at time $t$ and $R' \subseteq R$ be the set of released requests that precedes $u$ in the permutation. So, we know that permutation $\sigma$ serves first the set $R'$ of requests, then the request $u$ and after that, the rest of the requests (set $B=\big(R \setminus R'\big) \cup \big(U \setminus \{u\}\big)$). So, we can write its length $\ell_{\sigma}$ as
     \begin{equation*}
         \ell_{\sigma} = \ell_{\orig-R'-u} + \ell_{u-B-\orig}\enspace.
     \end{equation*}
     It is easy to see that the value $(1-\alpha_{\sigma})\ell_{\sigma}$ is exactly the length of the part of the tour starting with the unreleased request $u$, i.e., $\ell_{u-B-\orig}$. Similarly for its corresponding general dominating permutation $\sigma'$, we have that $(1-\alpha_{\sigma'})\ell_{\sigma'} = \ell_{TSP_{u-B-\orig}}$. It is also clear that $\ell_{TSP_{\orig-R'-u}} \le \ell_{TSP_{\orig-R'-u}}$ and $\ell_{TSP_{u-B-\orig}} \le \ell_{TSP_{u-B-\orig}}$. Thus, we get 
     \eqref{gen:length_1} and \eqref{gen:unreleased_1}.

     If $U  = \emptyset$, then $\sigma'$ is an optimal TSP tour that serves all  requests and returns to the origin and $\alpha_{\sigma}=\alpha_{\sigma'}=1$. Therefore, we have again that \eqref{gen:length_1} and \eqref{gen:unreleased_1} hold.
     
     Consequently, at each time $t$, a general dominating permutation $\sigma'$ dominates permutation $\sigma$.
 \end{proof}

Next, we give the formal definition of a general dominating set $D_i$.

\begin{definition}[General dominating set]
    We define the general dominating set $D_i$, with $i \in \{0, 1, 2, \dots, n\}$, as the set which contains all general dominating permutations that are generated with sets $R(t_i)$, $U(t_i)$ and every possible selection of $R'(t_i) \subseteq R(t_i)$ and $u \in U(t_i)$. 
\end{definition}

For the set $S$ of permutations $\swag$ uses each time, it holds that $S(t) \subseteq S(t')$ for every $t \le t'$, since we provide the algorithm at time $t_i$ with the union of all $D_k, k \le i \le n$. Moreover, the general dominating set $D_i$ dominates all permutations for all $t < t_{i+1}$. Thus, we will show that an oracle $\oracle$ which outputs at any time $t \in [t_i, t_{i+1})$ the general dominating sets 
$\bigcup_{j=0}^{i} D_j$ for each $i \in \{0, 1, 2, \dots, n\}$, is a domination oracle.
Before doing that, let us compute the running time of such an algorithm. 
First, we show that the number of general dominating permutations which are contained in the set $\bigcup_{j=0}^{i} D_j$ with $i\le n$, can always be upper bounded by $2^n$.

\begin{lemma} \label{run_time_dom}
There always exists a set $\bigcup_{j=0}^{i} D_j$ which consists of at most $2^n$ distinct general dominating permutations, for all $i \in \{0,1,\dots,n\}$.
\end{lemma}

\begin{proof}
Consider a general dominating permutation as it is defined in Definition~\ref{def:GenDomPer}. We only consider the closed variant here, the proof is analogous for the open variant. Suppose that we always take the permutation corresponding to the tour which consists of the (unique) lexicographically first optimal TSP paths $TSP_{\orig-R'-u}$ and $TSP_{u-B-\orig}$. Any such permutation is completely determined by the set $Q'=R'(t_i)\cup \{u\}$ for any value of $i$, as $u$ is the only unreleased request in $Q'$. Also, all sets $Q'$ are subsets of the set of all $n$ requests, and thus there can be at most $2^n$ distinct ones. Note that if $u$ does not exist, i.e., all requests have been released, there exists only a single lexicographically first optimal TSP tour.  Consequently, we can always construct the set $\bigcup_{j=0}^{i} D_j$ consisting of at most $2^n$ distinct (general dominating) permutations for all $i \in \{0,1,\dots,n\}$.
\end{proof}

Next, we use Lemma~\ref{run_time_dom} to show the following theorem for the running time of $\laswag$.

\begin{lemma}
$\laswag$ that uses the general dominating sets runs in $O(n^2\cdot 2^n)$ time complexity for both closed and open variants of OLTSP.
\end{lemma}

\begin{proof}
From Lemma~\ref{run_time_dom}, there are at most $N=2^n$ permutations in the set $S$ the algorithm uses. At $t=0$, we can precompute a set $D$ of general dominating permutations of size at most $2^n$ by Lemma~\ref{run_time_dom} in time $O(n^2 2^n)$ using dynamic programming~\cite{Bellman, Karp}. Then, every time $t_j$ a request is released, we can select the right set $D_j \subseteq D$ of general dominating permutations (from all the possible general dominating permutations we have precomputed), add it to the previous ones and provide them to the algorithm acting as the oracle. From Corollary~\ref{general:time-2}, we get that $\laswag$ with the general dominating sets runs in $O(n^2\cdot 2^n)$ time.
\end{proof}

Now we are ready to state and prove the main theorem of this section.

\expGeneral*

\begin{proof}
    If the two conditions of the domination oracle are satisfied for the general dominating sets, then we can use $\laswag$ and Lemma~\ref{dominating:general} to get an exponential time algorithm for open and closed LA-OLTSP that is $3/2$-consistent. It is obvious that Condition~$1$ is satisfied as $S(t) \subseteq S(t')$ for every $t \le t'$. Condition $2$ is also satisfied since for each time $t$ we have a corresponding dominating permutation for every possible subset of $R(t)$ and choice of $u$, and thus for every permutation $\sigma$. Therefore, we apply Lemma~\ref{dominating:general} and get a consistency of $3/2$. For smoothness and robustness, we use Theorems~\ref{th:smoothness},~\ref{robustness:2.75} and~\ref{robustness:2.84} directly, completing the proof.
\end{proof}

    \section{Polynomial/FPT time algorithms for specific metric spaces}\label{sec:polyTimeAlgos}

In this section, we show how to build polynomial/FPT time algorithms based on the general framework for specific metric spaces, namely trees, rings and flowers. Specifically, we show how to build (in polynomial/FPT time) a dominating set of polynomial/FPT size in each of these cases. These will dominate a so-called \textit{sensible} set of permutations in each case, which we define later. We focus now on the closed variant and explain later how to deal with the open case. Also, for the sake of simplicity, we will define the ``new'' dominating permutations that will be added to the oracle's set whenever a request is released. It is to be understood that the oracle maintains the union of all permutations generated and provides that to $\laswag$.

We also need to implement the ``cleanup'' step of $\laswag$, where a computation of an optimal classical (offline, without release times) TSP path/tour is required. We describe briefly how this is done in Subsection~\ref{subsec:offlineTSPfast}. Note, however, that these runtimes are small enough (in fact, in $O(n)$) so that the bottleneck will be the cardinality of the set of dominating permutations for the runtime of $\laswag$.

Now, let us start working towards the definitions and constructions promised.
\begin{definition}[Safe Set]
Let $X$ be a set of request locations. A set of permutations $\Pi$ is safe for $X$ if for any assignment $f$ of release times to each location in $X$, there exists a permutation $\pi \in \Pi$ which is optimal for the resulting input $Q = \{(x, f(x)) \: | \: x \in X\}$.
\end{definition}

Note that the set of all permutations is safe. However, as we will see, we can define sets of permutations that follow some desirable structure but are still safe. We will call such sets \textit{sensible}. It is these sets which we will dominate instead of the set of all permutations. Recall that $\laswag$ provides all the guarantees it does if the oracle that it uses dominates a safe set, which motivates this entire section. 

In general, we will show domination by the following steps. First, we will argue that the dominating permutation serves a superset of the requests served by the dominated permutation in no more time before reaching the first unreleased request $q$ (in the dominated permutation). Afterwards, an optimal solution that serves the rest of the requests will be followed. The above imply domination.

\subsection{Trees}

In this section, we consider metric spaces that are shaped like a tree. That is, we have continuous spaces that can be embedded into trees with lengths on the edges, where these are positive (and correspond to the traversal time of the edge) and the edges incident to a leaf can have ``infinite'' length. The request may appear on nodes of the tree or at some point along the edges. For example, a request may appear 1 unit away from $q$ and 2 units away from $p$, where $p$ and $q$ are nodes of the tree connected by an edge with length 3. This corresponds to requests along some street between crossings in a real-world scenario. It is important that we allow the leaf edges to have infinite length in order to be able to subsume the line and star metrics.

\paragraph*{Sensible set.} For such a metric space, we can focus on a very specific set of sensible permutations. Let us describe the idea behind the definition of these first. Suppose that an algorithm serves a request $q$ at time $t$. When the algorithm finishes, say at time $t_f$, it will be located at the origin $\orig$. There is only one path $\mathcal{P}_q$ from $q$ to $\orig$, and it must be traversed after $t$. Hence, any request $q'$ on $\mathcal{P}_q$ will be visited at some point after $t$ anyway. Therefore, there is no reason for the algorithm to try and serve $q'$ before time $t$ (i.e., before serving $q$), since it would serve $q'$ on the way back to the origin after $t$ anyway. From this, it follows that for any sequence of requests $S_q$ encountered along a path $\mathcal{P}_q$ towards the origin may be safely assumed to be a subsequence of the optimal permutation. This motivates the following.

\begin{definition}[Sensible Set for Trees]
    Let $\tree$ be a tree rooted at the origin $\orig$ and $X$ a set of request locations on $\tree$. The set $\Pi_{s}(\tree, X)$ of sensible permutations consists of all permutations $\pi$ where the following holds. Let $\pi = x_{\pi(1)}, x_{\pi(2)}, \dots, x_{\pi(n)}$. For any $x_{\pi(i)}$ and any $x_{\pi(j)}$ on the path $\mathcal{P}_{x_{\pi(i)}}$ from $x_{\pi(i)}$ to the origin, we have $j > i$.
\end{definition}

\begin{claim}
    For any tree $\tree$ and any set of request locations $X$, the set $\Pi_{s}(\tree, X)$ is safe.
\end{claim}

\begin{proof}
    Let $\pi$ be an arbitrary permutation and $f$ an arbitrary assignment of release times to $X$. We will show that there exists $\pi_{s} \in \Pi_{s}(\tree, X)$ with $|\pi_{s}| \le |\pi|$, where $|\pi|$ denotes the completion time of the tour associated to permutation $\pi$ under $f$. Let $z(t)$ denote the position of a server following $\pi$ at time $t$. Let $l_i = \max\{t \: | \: z(t) = x_i\}$. Now, we sort the $l_i$ values and extract the sorted indices to get $\pi_{s} = x_{\pi_{s}(1)}, x_{\pi_{s}(2)}, \dots , x_{\pi_{s}(n)}$. In other words, $\pi_s$ is derived by the ordering of the requests based on their last visit by a server following $\pi$. By construction, $\pi_s \in \Pi_{s}(\tree, X)$. Now, let us see why $|\pi_s| \le |\pi|$. Let $t^{run}$ be the final time that $\pi_s$ waits on a request location $x_{\pi_s(f)}$. It is easy to see that the release time of $x_{\pi_s(f)}$ is at least $t^{run}$. This means that $\pi$ still has to perform a superpath of $\mathcal{P}_{run} = x_{\pi_s(f)} \rightarrow x_{\pi_s(f + 1)} \dots \rightarrow x_{\pi_s(n)} \rightarrow \orig$, after time $t^{run}$. Let $l_{run}$ denote the length of $\mathcal{P}_{run}$. It is clear that $|\pi| \ge t^{run} + l_{run} = |\pi_s|$. 
\end{proof}

\paragraph*{Dominating set.} For the above sensible set, we will define a dominating set whose cardinality is of order $O(2^l \cdot n)$, where $l$ is the number of leaves of the underlying tree. This yields an algorithm for the problem in trees that runs in FPT time parameterized by the number of leaves. As the reader may have already guessed, the idea is to look at subsets of leaves and, more interestingly, at the paths of those leaves to the origin. That is, for each subset of the paths of the leaves to the origin and each selection of intermediate request $q$, we will define the permutation which first optimally visits the selected subset of leaves finishing at $q$ and then ``cleans up'' the rest of the requests starting from $q$ and ending at the origin. The intuition explaining why this is enough is that any sensible permutation can only have served a request $q'$ before $q$ if it already has served the leaves ``below'' $q'$. 

We can now formally define the dominating permutations. To do that, we consider the restriction of the tree space to the span of the requests. That is, for every leaf, if there is a request along its incident edge, we replace it with a new node which is the request furthest along the possibly infinitely long edge of the leaf and replace the infinite length by that request's distance from the original leaf's parent. Additionally, we ``trim'' the tree past requests when there is no other request along the path to a leaf. Now, every leaf is a request. 

\begin{definition}[$(q,L')$-Scan]\label{def:scan}
Let $\tree$ be a rooted tree and $L$ its set of leaves. Let $L'$ be a subset of $L$ and $q$ be some request on the tree with location $x$ such that $x$ is on some path to a leaf in $L'$ ($x$ could be along an edge). Now, let $\mathcal{P}$ be the optimal path starting at the root and finishing at $x$ while also visiting all leaves in $L'$ along the way. We call $\mathcal{P}$ the $(q, L')$-scan of $\tree$. 
\end{definition}

\begin{definition}[$(q,L')$-Dominator]\label{def:tree-dominator}
    Let $R, U$ be the set of released and unreleased requests at time $t$. Let $L \subseteq R$ be the subset of released leaves and also let $L'$ be a subset of those leaves. Let $q \in U$ be a furthest unreleased request along a path to some leaf. Let $\mathcal{P}'$ be the $(q, L')$-scan of $\tree$. Now, let $U'$ be the set of unserved requests after this scan. That is, $U' = U \cup R'$, where $R'$ is the subset of released requests not served by $\mathcal{P}'$. Construct the tree $\tree'$ rooted at $q$, trimmed according to $U'$ and preserving the origin. Let $L''$ be the leaves of $\tree'$. Let $\mathcal{P}''$ be the $(\orig, L'')$-scan of $\tree'$. The implied permutation $\pi$ of the tour $\mathcal{P}'\mathcal{P}''$ is the $(q, L')$-dominator of $\tree$. We will use the notation $d_{R, U, \tree}(q, L')$ to refer to it.
\end{definition}

\begin{definition}[Dominating Set for Trees]\label{def:tree-dominators}
    Let $R, U$ be as in definition \ref{def:tree-dominator} and $\tree, L$ as in definition \ref{def:scan}. The tree dominating set is:
    \begin{displaymath}
    Dom^{\curlyvee}(R, U, \tree, L) = \bigcup_{q \in U^*, L' \subseteq L}d_{R, U, \tree}(q, L'),
    \end{displaymath}
    where $U^*$ is the set of furthest unreleased requests along the paths to leaves.
\end{definition}

We will now show that this set does indeed dominate the sensible permutations for trees. Recall that a permutation is sensible for the tree if all sequences of requests along a single path from a leaf to the root (origin) ordered decreasingly by distance to the origin are found in the same order in the permutation. 

\begin{lemma}[Tree Domination]
Let $\pi$ be a permutation which is sensible for a tree $\tree$ (trimmed according to the input requests). Let $t$ be the current time and $R, U$ the released and unreleased requests at time $t$. Also, let $L$ be the set of leaves of $\tree$. Then $Dom^{\curlyvee}(R, U, \tree, L)$ contains a permutation $\pi_{dom}$ which dominates $\pi$.
\end{lemma}
\begin{proof}
Let $\pi = q_{\pi(1)}, q_{\pi(2)}, \dots, q_{\pi(n)}$. Let $q_{\pi(u)}$ be the first unreleased request in $\pi$. Let $R_{served} =  \{q_{\pi(1)}, q_{\pi(2)}, \dots, q_{\pi(u - 1)}\}$ be the requests served by $\pi$ before $q_{\pi(u)}$. Let $L' = L \cap R_{served}$. Then, the $(q_{\pi(u)}, L')$-dominator $\pi_{dom} = d_{R, U, \tree}(q_{\pi(u)}, L') \in Dom^{\curlyvee}(R, U, \tree, L)$ dominates $\pi$. Due to the construction of $\pi_{dom}$, specifically because an optimal path is calculated for the parts before and after $q_{\pi(u)}$, we need only show that $\pi_{dom}$ serves a superset of what $\pi$ serves before $q_{\pi(u)}$ in no more time. Note that because $\pi$ is sensible, $R_{served}$ must be closed under path-extension. That is, for every $q \in R_{served}$ and every path $\mathcal{P}$ that $q$ is on, all requests further down the path $\mathcal{P}$ than $q$ must also be in $R_{served}$. This implies that all requests in $R_{served}$ are along some path to a leaf in $L'$. But, $\pi_{dom}$ traverses all these paths before reaching $q_{\pi(u)}$, which proves that $\pi_{dom}$ serves a superset of $R_{served}$ before $q_{\pi(u)}$. It remains to show that it does so at least as fast as $\pi$ serves $R_{served}$. This follows because $\pi_{dom}$ is chosen to optimally serve $L'$ only and end up at $q_{\pi(u)}$, which $\pi$ also has to do anyway, completing the proof.
\end{proof}

Now, using the above lemma, we obtain our main theorem for trees. The complexity stems from the term related to the cardinality of the dominating set in Lemma \ref{general:time}, since solving classical TSP on trees is later shown to take $O(n)$ time.

\trees*

\subsection{Ring}


\paragraph*{Sensible set.} The ring turns out to be quite similar to the line. In fact, to define the sensible permutations for it, we simply need to make the following observation. If a permutation makes more than one ``full'' loops around the ring, it can drop all but the last, since any loop will move through every request. Since it does not make sense to serve anything before doing a full loop, the loop can be the first movement after leaving $\orig$. After such a loop, if it even exists, we can assume that the permutation moves along the ring as if it were a line, i.e., it never crosses the antipodal point of the origin. Hence, we can define the sensible permutations by branching on whether a loop is performed and then reduce to the tree case. 

\begin{definition}[Sensible Set for the Ring]
    Let $X$ be a set of request locations on the ring. The sensible set of permutations $\Pi_{s}^{\circ}(X)$ consists of all permutations $\pi$ resulting from the concatenation of $\pi_{loop}$ and $\pi_{line}$ where $\pi_{loop}$ is a subpermutation covering $X' \subseteq X$ in a cyclic fashion and $\pi_{line} \in \Pi_{s}(T, X \setminus X')$, where $T$ is the tree resulting from splitting the ring at the midpoint across from the origin.
\end{definition}

\begin{claim}
\label{Claim:Opt-Sensible}
    There is always an optimal solution which follows a sensible permutation.
\end{claim}

\begin{proof}
    We can assume that the ring is a circle with a circumference of $1$. Let us now consider an optimal solution $\opt$, and denote by $\sigma ^*$ the order in which $\opt$ serves the requests (a permutation of the requests). We recall that w.l.o.g.\ $\opt$ follows the tour/path $\sigma^*$, waiting only at requests' positions: it goes from $\orig$ to (the position of) $\sigma^*[1])$ in time $d(\orig,\sigma^*[1])$, waits at (the position of) $\sigma^*[1]$ if the request is not released, then from $\sigma^*[1]$ to $\sigma^*[2]$ in time $d(\sigma^*[1],\sigma^*[2])$, \dots  We will then prove that for every optimal solution $\opt$, there is another optimal solution $\opt'$ that is a sensible permutation.

    Consider the last time $\opt$ moves from the one semicircle $\widehat{\orig A\orig'}$ to the other one $\widehat{\orig B\orig'}$ (or vice versa) through the antipodal point $\orig'$ of the origin (Figure~\ref{fig:ring-antipodal}).
\begin{figure}[t]
    \centering
    \begin{tikzpicture}[scale=0.8]
	\draw (0,0) arc (270:-90:2.5);  
	\node[] at (0,-0.425) (0) {$\orig$};
    \draw (0,0) -- (0,-0.2);
    \draw [dashed] (0,0) -- ({2.5*cos(90)},{2.5+2.5*sin(90)});
    
	\node[] at ({2.95*cos(180) }, {2.5+2.95*sin(180) }) (0) {$A$};
     \draw ({2.5*cos(180)},{2.5+2.5*sin(180)}) -- ({2.7*cos(180)},{2.5+2.7*sin(180)});
    \node[] at ({2.95*cos(150) }, {2.8+2.95*sin(150) }) (0) {$\sigma^*[i_{l}]$};
    \draw ({2.5*cos(150)},{2.5+2.5*sin(150)}) -- ({2.7*cos(150)},{2.5+2.7*sin(150)});
    \node[] at ({2.95*cos(90) }, {2.5+2.95*sin(90) }) (0) {$\orig'$};
    \draw ({2.5*cos(90)},{2.5+2.5*sin(90)}) -- ({2.7*cos(90)},{2.5+2.7*sin(90)});
    \node[] at ({3.2*cos(-15) }, {5+3.2*sin(-15) }) (0) {$\sigma^*[i_{l}+1]$};
    \draw ({2.5*cos(30)},{2.5+2.5*sin(30)}) -- ({2.7*cos(30)},{2.5+2.7*sin(30)});
    \node[] at ({2.95*cos(0) }, {2.5+2.95*sin(0) }) (0) {$B$};
    \draw ({2.5*cos(0)},{2.5+2.5*sin(0)}) -- ({2.7*cos(0)},{2.5+2.7*sin(0)});
    \end{tikzpicture}
    \caption{Ring for the proof of Claim~\ref{Claim:Opt-Sensible}.}
    \label{fig:ring-antipodal}
\end{figure}
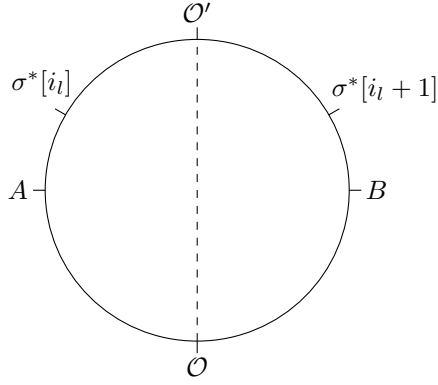
    If $\opt$ doesn't cross $\orig'$ in any moment, then $\opt'$ can just follow $\opt$ which has a zig-zag motion and is a sensible permutation. Otherwise, there is a pair of consecutive requests $\sigma^*[i_{l}]$, $\sigma^*[i_{l}+1]$ in the permutation between which the crossing of $\orig'$ occurs. Assume w.l.o.g. that $\sigma^*[i_{l}] \in \widehat{\orig A\orig'}$ and $\sigma^*[i_{l}+1] \in \widehat{\orig B\orig'}$.
    Then, $\opt$ has to move from $\sigma^*[i_{l}+1]$ to the origin without crossing $\orig'$ again. Let $t^*$ be the first time that $\opt$ reaches the origin after serving the request at $\sigma^*[i_{l}+1]$. Instead of following $\opt$, $\opt'$ can wait in the origin for time $t^*-1$ and then make a full circle $\widehat{\orig A\orig'B\orig}$ without stopping anywhere. In this case, it reaches $\orig$ at time $t^*$ (the same as $\opt$) serving at least the requests $\opt$ has served until that time. Next, $\opt'$ just follows $\opt$ which has a zig-zag motion since it never crosses $\orig'$ again. Therefore, $\opt'$ is a sensible permutation and $|\opt'| \leq |\opt|$.
    
\end{proof}

\paragraph*{Dominating set.}
As we have discussed, a sensible permutation for the ring can have at most one full loop in the beginning. Note that if this loop is missing, the case degenerates to the line, for which we have already defined the dominators. Therefore, it only remains to define the dominators under the assumption that $\pi$ does indeed have a loop in the beginning, since then we can just take the union of the two sets to construct the set of all dominators.

A crucial distinction is whether the current request $q$ of $\pi$ is part of the loop or not. In the first case, $\pi$ has only traveled a portion of the loop so far and in the second case it has traveled a full loop and then taken some possibly complex line-like tour. For every choice of $q$, we define two permutations, one pertaining to the first case and another to the second one.

The first permutation serves all released requests within $\pi$'s traveled portion of the loop (in the same direction) and then moves to $q$. From there, it follows an optimal permutation serving the rest of the requests  (disregarding release times), which can easily be computed for the ring. The second permutation serves all released requests with a full loop (in clockwise order if $q$ is to the left arc of the ring and vice versa) and then moves to $q$. It finishes in the same way by computing an optimal path from $q$ to the origin. We now define this set of dominating permutations more formally.

\begin{definition}[Crescent Permutations]\label{def:crescents}
Let $R$ be a set of released requests and $U$ be the corresponding set of unreleased requests on the ring at time $t$. Let $r_1, \dots, r_m$ be the released requests in clockwise order from the origin. For every $q \in U$, the \textbf{left-crescent permutation} of $q$ is
\begin{displaymath}
(r_1, \dots, r_l)\ ||\ (q)\ ||\ opt_{q}((R \cup U) \setminus \{r_1, \dots, r_l\}),
\end{displaymath}
where $l$ is the largest index such that $r_l$ is to the left of $q$ in clockwise order from the origin and $opt_{s}(S)$ denotes the fastest path without release times that starts from $s$, visits every point in $S$ and ends at the origin. The \textbf{right-crescent permutation} of $q$ is defined analogously. We will use $lc_{R, U}(q)$ and $rc_{R, U}(q)$, respectively, to refer to these permutations.
\end{definition}

\begin{definition}[Full-Moon Permutation]\label{def:full-moon}
Let $R, U$ be as in Definition~\ref{def:crescents} and $q \in U$. If $q$ is on the left arc of the ring, i.e., it is encountered before the antipodal point of the origin in a clockwise traversal of the ring, then the full-moon permutation of $q$ is 
\begin{displaymath}
    (r_1, \dots, r_m)\ ||\ (q)\ ||\ opt_{q}(U).
\end{displaymath}
If $q$ is on the other arc, the order of serving the released requests is reversed. We will use $fm_{R, U}(q)$ to refer to this permutation.
\end{definition}

\begin{definition}[Dominating Set for the Ring]\label{def:ring-dominators}
    For given $R, U$ as in Definitions~\ref{def:crescents} and~\ref{def:full-moon}, the ring dominating set of $R, U$ is
    \begin{displaymath}
         Dom^{\circ}(R, U) = \bigcup_{q \in U} \{rc_{R, U}(q), lc_{R, U}(q), fm_{R, U}(q)\}.
    \end{displaymath}
\end{definition}

Now that we have defined the dominators, it remains to show that indeed they dominate the sensible set of permutations which contain a loop. 

\begin{lemma}[Ring Domination]
Let $\pi$ be a sensible permutation whose tour contains a full loop in the beginning. Let $t$ be the current time and $R, U$ the released and unreleased requests at time $t$. Then, $Dom^{\circ}(R, U)$ contains a permutation $\pi_{dom}$ which dominates $\pi$.
\end{lemma}

\begin{proof}
Let $\pi = q_1, q_2, \dots, q_n$. Assume without loss of generality that $\pi$'s loop is clockwise. Also, let $q_x \in U$ be the first unreleased request in $\pi$ at time $t$. We distinguish two cases:

\begin{itemize}
    \item The request $q_x$ is on the loop part of $\pi$. Then, $lc_{R, U}(q_x) \in Dom^{\circ}(R, U)$ is a dominator for $\pi$. This is because both permutations travel the same distance to $q_x$ (the clockwise arc to it) and $lc_{R, U}(q_x)$ serves a superset (all released request in that arc) of what $\pi$ serves (some released requests of that arc) until this point. Also, after reaching $q_x$, $lc_{R, U}(q_x)$ optimally serves a subset of what $\pi$ has to serve, and is thus faster on that part too, which means that it dominates $\pi$.
    \item The request $q_x$ is after the loop part of $\pi$. Then, $fm_{R, U}(q_x) \in Dom^{\circ}(R, U)$ is a dominator for $\pi$. This can be seen because $fm_{R, U}(q_x)$ serves all released requests in its loop part and moves to $q_x$ after the loop at the least possible time (with cost 1 for the loop and then the smallest arc to $q_x$). Thus, it serves a superset of what $\pi$ serves until reaching $q_x$ in no more time than the tour induced by $\pi$ needs. Then, it optimally serves a subset of what $\pi$ has to serve, and is thus faster on that part too, which means that it dominates $\pi$.
\end{itemize}
\end{proof}

Thus, our main theorem for the ring follows.

\ring*
\clearpage
\subsection{Flowers}
\begin{wrapfigure}{r}{2cm}
\includegraphics[width=2cm]{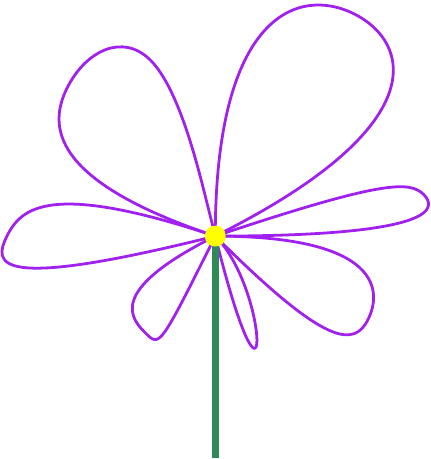}
\caption{\phantom{A} A flower with 7 petals}\label{wrap-fig:1}
\label{fig:flower}
\end{wrapfigure}
In this section we extend the reasoning used in the previous subsections on rings and trees to the so-called flowers. A flower consists of a number of rings (petals), all of which are attached to the origin point (receptacle). For reasons of artistic completion, a semiline (stem) disjoint from the petals is also attached to the origin point\footnote{In fact, one could do the same for graphs which are trees plus some edge-disjoint cycles attached to a single vertex of the tree. Thus, one can even include the flower's (ground) roots.}. An illustration is given in Figure~\ref{fig:flower}. In the following, we define the sets of sensible and dominating permutations for these kinds of graphs.

\paragraph*{Sensible set.}
The central idea is the following. We can assume that any sensible solution traverses a petal of the flower in a cyclic fashion at most once, for the same reason that this is true for the ring. Moreover, the petal is not visited before without loss of generality. After such a traversal (if it ever occurs), we can safely ``snip'' the petal in its halfpoint from the origin and view it as two different branches of a tree, since the halfpoint will never be crossed again. Thus, the sensible set here is defined to include the permutations which have at most one loop for each petal, such loops are the first visit on the corresponding petal (but may be carried out fully during any contiguous part of the permutation) and the restrictions of the sensible set on trees also apply to every petal (which becomes a line) after such a loop (if it even exists) is carried out.

\paragraph*{Dominating set.}
Utilizing this observation, we can construct the dominating permutations as follows. First, choose an unreleased request $q$. This will be the ``current'' request of the sensible permutation for which we want to define the dominator. Let $\pi$ refer to that permutation. Note that $\pi$, throughout its entirety, will utilize a subset of the petals in a cyclic fashion. So, we ``guess'' this subset and call it $\mathcal{P}$. Moreover, a subset $\mathcal{P}_{done}$ of $\mathcal{P}$ will have been fully traversed before reaching $q$. We also guess $\mathcal{P}_{done}$. This also defines $\mathcal{P}_{left} = \mathcal{P}\setminus \mathcal{P}_{done}$. Then, the dominator is defined in the following way. First, we traverse all petals in $\mathcal{P}_{done}$ in an arbitrary order. Next, we ``snip'' every petal except the ones in $\mathcal{P}$. This leaves us with a tree $\tree$ (actually, a star) and the petals in $\mathcal{P}$, which are disjoint from each other. Note that this is safe to do if we correctly guessed the sets of petals, since $\pi$ will behave as if it is on a tree when visiting the petals not in $\mathcal{P}$.

Now, we also guess the subset $L'$ of the leaves $L$ of $\tree$ to be visited before $q$. Having fixed that, we append to our initial petal loop subtour the optimal path from $\orig$ to $q$ that visits every leaf in $L'$, and also possibly serves the arc that $\pi$ might have served if it is currently performing a loop on $q$'s petal. Finally, we append the optimal path from $q$ that ``cleans up'' the remaining unserved requests.

Note that to make all these different guesses, one has at most six choices for every petal (loop done, loop later, or any of the 4 choices of subsets of the two leaves after snipping). So, if the graph has $p$ petals, the cardinality of the dominators is $O(6^p \cdot n) = O(2^{2.59p} \cdot n)$. The above calculation considers the stem as a petal for simplicity. Again, this yields an FPT algorithm for the flower case.

Here we give a proof sketch of why this works. Recall that we only need to show that any sensible permutation $\pi$ is dominated at all times by some dominator $\pi_{dom}$ out of the above. So, for $\pi$ at time $t$, we claim that $\pi_{dom}$ which corresponds to the correct guess of $q, \mathcal{P}, \mathcal{P}_{done}$ and $L'$ dominates $\pi$. Let us first focus on the petals of $\mathcal{P}_{done}$. It is clear that $\pi_{dom}$ serves all released requests in these petals in the least time necessary to traverse the petals in loops, an amount of time which $\pi$ also pays by the correct guess of $\mathcal{P}_{done}$. Now, we shift our focus into the part of $\pi_{dom}$ after the initial loops and until reaching $q$. By construction, it is the optimal path visiting $L'$ and any possible requests that $\pi$ might have served if currently on a loop traversal of $q$'s petal. But, by a similar reasoning as in the trees subsection (the path-extension property of the served requests of $\pi$), $\pi$ also has to accomplish this task, independently of visiting $\mathcal{P}_{done}$ in loops. This independence is because $L'$ has no leaves in petals of $\mathcal{P}_{done}$, since all released requests there are served in the initial loops of $\pi_{dom}$ and also because in the case that $q$ is part of a loop, its petal is not in $\mathcal{P}_{done}$. Moreover, $\pi_{dom}$ serves all requests that $\pi$ has served until reaching $q$, since each request that $\pi$ has served is either in a petal of $\mathcal{P}_{done}$, along a path to a leaf in $L'$, or along the (correctly chosen) arc to $q$. Hence, $\pi_{dom}$ serves a superset of what $\pi$ serves and reaches $q$ at least as fast. Domination now follows by the optimality of the ``cleanup'' path.

We will now make the above more formal by precisely defining the dominators and proving that they dominate the sensible set of permutations.

\begin{definition}[Hippie Permutations]\label{def:flower_dominator}
Let $R, U$ be the sets of released requests and unreleased requests at time $t$. Let $q \in U$ be an unreleased ``current'' request. Let $\mathcal{P}$ be a subset of the petals of the input flower metric space. Let $\mathcal{P}_{done} \subseteq \mathcal{P}$. For now, let us assume that $q$ is not on a petal of $\mathcal{P}$. Let $\tree$ be the tree obtained by replacing the petals \textbf{not in} $\mathcal{P}$ by two ``branches'' each with half the length of the original petal (snip the petal at its halfpoint). $\tree$ possibly also contains the stem of the flower. Finally, let $L$ be the leaves of $\tree$ (trimmed according to the released requests) and $L' \subseteq L$. The hippie permutations $hp_{R, U}(q, \mathcal{P}, \mathcal{P}_{done}, L')$ are defined as follows:

\begin{itemize}
    \item First, all released requests in the petals of $\mathcal{P}_{done}$ are served by looping through each petal once in an arbitrary way, e.g.\ clockwise and by prioritizing the petals according to some fixed ordering.
    \item Then, an optimal path is calculated which visits each leaf in $L'$ and ends up at $q$. We append this path (more accurately, the subpermutation associated to it) to the permutation.
    \item Finally, the optimal path which starts at $q$ and visits any unserved until now request before finishing at the origin is calculated and appended.
\end{itemize}

In the case that $q$ is in a petal of $\mathcal{P}$, we slightly modify the above definition by requiring that in the second part, not only leaves in $L'$ are visited but also all requests along a petal arc to $q$. We define two such permutations, one for each arc leading to $q$. Thus, $hp_{R, U}(q, \mathcal{P}, \mathcal{P}_{done}, L')$ will refer to the set of permutations generated in the above way (one or two permutations actually).
\end{definition}

\begin{definition}[Dominating Set for Flowers]\label{def:flower_dominators}
    Let us reuse the notation of Definition~\ref{def:flower_dominator}. We will describe now how the dominating set for flowers is generated. We simply take the union of the hippie permutations for all (valid) choices of $\mathcal{P}, \mathcal{P}_{done}, q, L'$. That is, we define the set:

    \begin{displaymath}
    Dom^{\text{\FiveFlowerPetal}}(R, U) = \bigcup_{\substack{q \in U\\
    \text{ $\mathcal{P}$ is a subset of petals}, \\
    \mathcal{P}_{done} \subseteq \mathcal{P}, \\
    \text{$L'$ is a subset of the leaves of $\tree$}}} hp_{R, U}(q, \mathcal{P}, \mathcal{P}_{done}, L')
    \end{displaymath}
\end{definition}

Finally, we can now show that the above set dominates the set of sensible permutations for the flower.

\begin{lemma}[Flower Domination]
Let $\pi$ be a sensible permutation for the flower and let $R, U$ be the released and unreleased requests at time $t$. Then, $Dom^{\text{\FiveFlowerPetal}}(R, U)$ contains a permutation $\pi_{dom}$ which dominates $\pi$.
\end{lemma}
\begin{proof}
Let $q$ be the first unreleased request in $\pi$ at time $t$. Also, let $\mathcal{P}$ be the set of petals that $\pi$ will loop through (or have already done so) at some point. Let $\mathcal{P}_{done} \subseteq \mathcal{P}$ be the subset of those which have already seen their loop completed before reaching $q$. Let $\tree$ be as in Definition \ref{def:flower_dominator}. Let $L'$ be the subset of leaves that $\pi$ visits before $q$. We claim that $hp_{R, U}(q, \mathcal{P}, \mathcal{P}_{done}, L')$ contains the permutation $\pi_{dom}$ we are looking for.

\begin{itemize}
    \item Let us assume for now that $q$ is not on a petal of $\mathcal{P}$. Then, $hp_{R, U}(q, \mathcal{P}, \mathcal{P}_{done}, L')$ is the singleton $\{\pi_{dom}\}$. To see why, note that we can split the ``movement'' of $\pi_{dom}$ up to $q$ into two independent parts, one for the initial loops and one for the rest. The first part costs as much as the total length of the petals in $\mathcal{P}_{done}$, which $\pi$ also has to pay since it has fully looped through these petals at some point before reaching $q$. Now, notice that none of the leaves of $L'$ are in petals of $\mathcal{P}_{done}$, by construction. Therefore, $\pi_{dom}$ optimally visits these leaves and $\pi$ also has to additionally pay at least this cost (on top of and independently of the cost associated with the loops of $\mathcal{P}_{done}$), before reaching $q$. Therefore, we have shown that $\pi_{dom}$ follows a not longer path until $q$. It only remains to show that with that path it serves every request that $\pi$ has served until $q$. For the loop petals, this is obvious, since $\pi_{dom}$ serves every released request there. Then, for all requests outside of $\mathcal{P}$, we can view the problem as if we are on a tree. Thus, the arguments of the tree subsection apply (closure under path-extension etc.). It remains to show that there is no request in $\mathcal{P}_{left}$ which $\pi$ has served but $\pi_{dom}$ has not. This is trivial, since there is no such request because of the assumption that $q$ is not on a petal of $\mathcal{P}_{left}$.

    \item Now we can resolve the case where $q$ is on a loop petal. This means that $\pi$ is currently undergoing a loop. We choose $\pi_{dom}$ to be the permutation following the corresponding arc to $q$. The fact that $\pi_{dom}$ pays a not higher cost until reaching $q$ can be established in much the same way (we split the costs into the first loops and then $\tree$ plus the final arc). The only extra thing we need to show is that $\pi_{dom}$ serves every request in $\mathcal{P}_{left}$ that $\pi$ does until $q$. But these are just the requests in the correctly chosen arc to $q$, and thus the claim follows.
\end{itemize}
\end{proof}

Again, the lemma above establishes our main theorem for flowers. As for the trees, the complexity comes from the term related to the cardinality of the dominating set in Lemma \ref{general:time}, since solving classical TSP on flowers is later shown to take $O(n)$ time.

\flowers*
\subsection{Dealing with the open variant}
Much of the above analysis was based on exploiting the closed nature of the problem, specifically the fact that $\opt$ is required to return to $\orig$. This is no longer true in the open case. However, we can artificially impose a similar restriction by enumerating all possible endpoints of $\opt$'s path. That is, we will assume that a specific request $q_f$ will be served last by $\opt$ and build the dominating permutations under this assumption. In the end, we consider the union of all these sets (i.e., for all choices of the final request), which will dominate the sensible sets for the open case, since any sensible permutation has to finish with \textit{some} request. 
\paragraph*{Trees.} For the case of trees, this turns out to be quite simple. Once we fix $q_f$, we simply need to reroot the tree at $q_f$. Then, all sensible permutations which terminate at $q_f$ must respect the order of requests along paths to leaves of the rerooted tree; a simplification very similar to the one we had for the closed case. Hence, we can build the set of dominating permutations by considering the union of tree dominators for all $n$ rerootings of the input tree. Note that rerooting cannot increase the number of leaves by more than 1. The cases where this happens occur only when the new root is actually along an edge, and then we have to split that edge into two. This set dominates the sensible set of permutations for the open case in trees and has size $O(2^l \cdot n^2)$, yielding an FPT algorithm once again.
\paragraph*{Ring.}
For the case of the ring, we take a different approach altogether compared to the closed variant. Instead of specifying the sensible set and then dominating it, we simply dominate the set of all permutations with a relatively simple case distinction.
\begin{itemize}
    \item The first easy case is when the permutation $\pi$ to be dominated has not crossed the point $\orig'$, the antipodal point of the origin before reaching the current request $q$. In such a scenario, we view the ring as a line and the already defined dominators for trees suffice. The correspondence here would be that a permutation $\pi$ is dominated by the dominator $\pi_{dom}$ which visits the same set of extreme requests (leftmost or rightmost) as $\pi$ before $q$.
    \item Now, let us assume that $\pi$ has indeed crossed the point $\orig'$ before reaching $x$, the location of $q$. Let us be rid of the case where $\pi$ has moved a distance of at least $1 + d(x, \orig)$ before reaching request $q$, as then the permutation that serves all released requests with a loop and then goes to $x$ as soon as possible dominates $\pi$. Let us assume w.l.o.g.\ that between the last touch of $\orig$ and the first subsequent touch of $\orig'$, $\pi$ moves through the right arc.

    Now, let us focus on the left arc. Let $q_1$ be the request closest to $\orig'$ on the left arc that $\pi$ has served before its last visit of $\orig$. Let $q_2$ be the request closest to $\orig$ on the left arc that $\pi$ has served after the last visit of $\orig$ and  before leaving $q$ (it could happen that $q_2 = q$). By $x_1$ and $x_2$, we denote their location, respectively. Because these points have to be traveled while moving a distance not more than $1 + d(x, \orig)$, we deduce that $d(x_1, \orig) < d(x_2, \orig)$. Note that no request in the left arc of $(x_1, x_2)$ could have been served by $\pi$ before request $q$. 
    
    For such $\pi$, we define the following dominator. First, $\pi_{dom}$ moves to $x_1$ and serves all requests on the way back to the origin. Then, it serves all requests along the right arc to $\orig'$, in that order. After reaching $\orig'$, it continues moving towards $x_2$, serving all requests on the way along the left arc. Finally, it moves to $x$, after which point it optimally cleans up the still unserved requests. We see that $\pi_{dom}$ serves all requests outside the left arc of $(x_1, x_2)$, i.e., it serves a superset of what $\pi$ serves, before serving $q$. Also, the movements of $\pi_{dom}$ can be charged independently to movements of $\pi$, and so the distance traveled is not higher, which implies domination.

    The majority of the above dominators are of the final type where we have to choose $q, q_1$ and $q_2$ (we do not need to choose $q_f$ as in the case of trees). Hence, there are $O(n^3)$ dominators, which yields a polynomial time algorithm.
\end{itemize}

\paragraph*{Flowers.}
For non-trivial flowers (i.e., not just a single ring), the situation is a bit more complex, but we can still retrieve an FPT algorithm. Note that the final request $q_f$ will either lie on a petal or not. If not, the argument that \textit{all} petals may be traversed in a cyclic fashion at most once (and not visited at all before such a loop) still holds. Thus, simply rerooting the tree and guessing the parameters of the closed variant flower dominators for the modified tree (but afterwards appending the optimal solution to the open variant) still works.

Now, if $q_f$ happens to lie on a petal, we have some extra work to do. Note however that in this case, any path from a request $q$ to $q_f$ where $q$ is \textit{not} on the same petal as $q_f$ \textit{must} pass through the origin. Thus, we can still assume that the sensible permutations respect the order of requests along paths to leaves of the (possibly snipped) tree rooted at the origin, at least as far as requests not on the final petal are concerned. We do not reroot the tree in such cases. Plus, all non-final petals may be traversed cyclically at most once. We thus only have to deal with the requests on the same petal as $q_f$. Fortunately, these can be dealt with relatively simply, essentially by viewing the final petal as a single ring in the open variant.

Suppose the ``current'' request $q$ is not on the final petal. Then, we can deal with the final petal with either a loop or snipping it and guessing the leaves, as we would originally do. This is because before reaching $q$, $\pi$ will have either visited the final petal's midpoint or not. In the first case, it must have made a full loop, since $q$ is not on that petal. In the other, we view the final petal as a tree.

The remaining case is the interesting one, where both $q$ and $q_f$ are on a single petal. In such a case, we first guess the \textit{other} petals and tree leaves and serve them optimally in a closed manner, i.e., we return to the origin. Now, we ignore everything but the petal of $q$ and $q_f$ and append a ring-type dominator for that specific petal.

Overall, we get a dominating set whose cardinality is $O(n^3 \cdot 6^p) = O(n^3\cdot 2^{2.59p})$, where $p$ is the number of petals of the flower. This is because we have roughly a quadratic multiplicative factor more dominators over the closed case, since we additionally choose \textit{either} $q_f$ or $q_1$ and $q_2$. The first case corresponds to when $q_f$ is on the tree part (in which case we do not need to choose $q_1$ and $q_2$) and in the second case the choices of $q_1$ and $q_2$ imply the choice of petal on which $q_f$ is, which is all we need to know to define the ring-type dominator we append.

\subsection{Solving classical TSP fast}\label{subsec:offlineTSPfast}
Here we describe how we compute an optimal classical TSP solution within suitable runtime bounds. This is used in the ``cleanup'' step of Algorithm \ref{algo:generalSmoothRobust}, but also in the definition of the dominating permutations.

For trees, we can simply perform a DFS traversal (from the startpoint) while deprioritizing the edges that lead to $q$, where $q$ is the endpoint of the optimal path we would like to calculate. The corresponding path traced by the DFS is optimal with cost twice the sum of edge lengths minus the length of $q$'s path to the startpoint. Note that for our intents and purposes, we can assume that DFS runs in $O(n)$ time, where $n$ is the number of requests. To see this, we need to show that we can bound the size of the actual tree by a constant multiple of $n$, since then the complexity of DFS boils down to $O(n)$. As a first preprocessing step and only once, we do the following. We trim the tree so that only the paths from the requests to the root (can be any point, might as well be the origin) are retained. This means that we have $l \le n$, where $l$ is the number of leaves of the modified tree. Moreover, we assume that there exists no vertex of the tree with degree exactly 2 (except possibly the origin), as these are non-essential for the OLTSP problem (one can view the two incident edges as one). Now, notice that any vertex of the tree has degree at least 3, except for the leaves. Let $i$ be the number of non-leaf vertices of the tree. By the handshaking lemma, we have $2e = \sum_{\text{$v$}}\deg(v) \ge 3i + l \implies i \le l$, where $e$ is the number of edges of the tree and by recalling that $e = i + l - 1$. Hence, the tree does indeed have size $O(n)$, which shows that we can solve classical TSP in $O(n)$ time.

For flowers (which subsume the ring), we first conclude which petals are better traversed as a cycle or as two different branches. This is easy to decide by comparing the cost of the full cycle versus a zig-zag type traversal that visits all requests. Afterwards, we run a DFS on the part of the flower that is decided to be traversed as a tree and append the petal cycles at any part of said DFS traversal where the origin is touched. The case where no such part exists is degenerate and easily dealt with (just the semiline is traversed by the DFS). It is easy to see that this whole procedure runs in $O(n)$ time as well, provided we have discarded petals empty from requests as a preprocessing step.

We note that these overheads do not matter asymptotically for $\laswag$, since the bottleneck of the runtime already comes from the cardinality of the dominating sets times $n^2$, whereas for the computation of the dominating permutations we only pay $O(n)$ for each and there are up to $n+1$ calls to the domination oracle overall.

    \section{Conclusion}
We studied Online TSP augmented with predictions regarding the locations of the requests. Our algorithm, $\laswag$, achieves a competitive ratio of $3/2$ under the assumption of perfect predictions, which is tight in most cases we considered. Additionally, it is smooth and provides robustness guarantees below 3, improving over previous work. The runtime of $\laswag$ is single-exponential; however, we show how to remove the exponential dependency on the number of requests for specific metric spaces. 

We believe that our techniques can be generalized to obtain FPT algorithms for other classes of graphs also; cactus graphs, graphs of bounded treewidth in general, as well as planar graphs are interesting options. Additionally, extending the algorithm to the Dial-A-Ride Problem seems like a reasonable direction to follow. Finally, having slightly improved the consistency lower bound in the line open variant, we conjecture that it is possible to extend our construction and establish tightness in this case also; we suggest this as future work.

Another interesting direction is to leverage the ideas behind the PTAS for classical TSP on the Euclidean plane \cite{Arora96, Arora97} (or any other space which admits an approximation scheme for that matter) to obtain consistency guarantees which approach 3/2 arbitrarily close as the computational time is allowed to increase. Note that, since our smoothness proof follows from consistency itself and robustness from solving a classical TSP instance, these results would automatically extend, modulo some worsening of the bounds as a function of the approximation quality. For this, one possible approach is to extend the notion of a sensible set to that of $\epsilon$-sensible, meaning that a $(1 + \epsilon)$-approximation of the optimal Online TSP solution is guaranteed to exist within the set. A similar relaxation would also make sense for the idea of dominating sets.

\clearpage
\bibliography{bibliography}
\clearpage

\appendix

\section{Consistency lower bound on the line for the open case}

\begin{proposition} \label{open_lb_line}
No algorithm can have a better consistency than $\frac{1+\sqrt{61}}{6}$ for open LA-OLTSP on the line.
\end{proposition}
\begin{proof}
Let $\lambda=\frac{1+\sqrt{61}}{6}\approx 1.468$, which satisfies $\lambda=\frac{5}{3\lambda-1}$. 

Let $A$, $O$, $B$ be the points on the line with coordinates $x_A=-1$, $x_O=0$ and $x_B=1$. We consider that we have a continuous set of requests between $A$ and $B$ (this can be easily discretized to get the same asymptotic lower bound with a finite number of request). None are released until $t=1$. Let also denote $\delta=5-3\lambda$ (note that $\delta$ is slightly larger than 0.5).

We consider two cases at $t=1$. Let $s$ denote the position of the algorithm at $t=1$. We assume w.l.o.g.\ that $s\geq 0$.

    
If $s\geq 1-\delta$, then we release all the requests at unit speed, starting at $A$ at $t=1$ till $B$ at $t=3$. Then $|OPT|=3$. If the algorithm serves $B$ before $A$, it needs at least 5. If it serves $A$ before $B$, it needs at least $1+(1-\delta)+1+2=3\lambda$. Hence, the competitive ratio is at least $\lambda$.

Otherwise, $0\leq s<1-\delta$. Then we release the requests from $A$ and from $B$ at unit speed, until the first time $t_0$ such that the algorithm is at some position $s'$ with $|s'|=(2-t_0)-\delta$, i.e., the server is at distance $\delta$ from the ``released part'' of the requests. We distinguish two cases, depending on $t_0$. We assume w.l.o.g.\ that $s'\geq 0$ at $t_0$, so $s'=(2-t_0)-\delta$. Let $C$ and $D$ be the points at positions $x_C=-(2-t_0)$ and $x_D=2-t_0$, respectively (i.e., at $t_0$ the requests are released on $[A,C]$ and on $[D,B]$).
    \begin{itemize}
        \item If $t_0\geq 3\lambda-3$ (intuitively slightly less than half of the requests are released, and $\alg$ is close to $O$), then we release the requests from $C$ to $D$ at unit speed, starting at $t_0$. Then again $|OPT|=3$. If $\alg$ serves $A$ before $B$, as previously it needs $t_0+(2-t_0-\delta)+3=5-\delta=3\lambda$. If $\alg$ serves $B$ before $A$, it needs $t_0+(1-(2-t_0-\delta))+2=2t_0+\delta+1=2t_0+6-3\lambda\geq 3\lambda$. The competitive ratio is at least $\lambda$.
        \item Finally, if $t_0\leq 3\lambda -3=2-\delta$ (intuitively, $\alg$ is not ``too close'' to the origin): we release at $t_0$ all the requests but the one in $D$. Then intuitively either $\alg$ went towards $A$, and we will act as previously. Or it went towards $B$, and we will block the release in $D$. More formally, at $t=t_0+1$:
        \begin{itemize}
        	\item if the algorithm is at a position $s''<0$, then we release the last request (in $D$). Note that $\alg$ cannot have served $A$ (as we had $s'>0$ at $t_0$) nor $B$ (no time to have $s''<0$ at $t_0+1$). Then, the best $\alg$ can do after $t=t_0+1$ is to go to $A$ (as $s''<0$), and then to $B$. As we had $s'=(2-t_0)-\delta$ at $t_0$, we have $|\alg|\geq  t_0+(2-t_0)-\delta+3=5-\delta=3\lambda$, while $|\opt|=3$, so the competitive ratio is at least $\lambda$.
        	\item if $s''\geq 0$, we release the last request in $D$ at $t_1=3+(1-x_D)=2+t_0$. Then $|\opt|=t_1$ (be in $A$ at time 1, then go to $B$ (time 3) and back to $D$). If $\alg$ serves $A$ before $D$, as $s''\geq 0$ at $t_0+1$, it needs at least $t_0+1+2+x_D=5$. If it serves $D$ before $A$, as $D$ is released at $t_1$, it needs at least $t_1+x_D+1=5$. Hence, the competitive ratio is at least $\frac{5}{t_1}=\frac{5}{t_0+2}\geq \frac{5}{3\lambda-1}=\lambda$ (by definition of $\lambda$).
		\end{itemize}        
 
    \end{itemize}
\end{proof}

\end{document}